%% file: acda2025.tex
\documentclass[twoside,leqno,twocolumn]{article}
\usepackage[letterpaper]{geometry}
\usepackage{siamproceedings}
\usepackage[T1]{fontenc}
\usepackage{amsfonts}
\usepackage{epstopdf}
\usepackage{enumitem}
\ifpdf
  \DeclareGraphicsExtensions{.eps,.pdf,.png,.jpg}
\else
  \DeclareGraphicsExtensions{.eps}
\fi
\usepackage{amssymb}
\usepackage{hyperref}
\usepackage[dvipsnames]{pstricks}
\usepackage{pstricks-add}
\usepackage{epsfig}
\usepackage{pst-grad} 
\usepackage{pst-plot} 
\usepackage[space]{grffile} 
\usepackage{etoolbox} 
\usepackage{lipsum} 
\usepackage{mathtools}
\usepackage{algpseudocode}
\usepackage{listings}
\usepackage{algorithm}
\usepackage{graphicx}
\usepackage{tikz}
\usetikzlibrary{shapes,arrows,automata,decorations.pathreplacing,angles,quotes}
\usepackage{bm}
\usepackage{float}
\usepackage{diagbox}
\usepackage[pdf]{graphviz}
\usepackage{siunitx}
\usepackage{placeins}
\usepackage{dblfloatfix}
\usepackage{caption}
\newsiamremark{remark}{Remark}
\newsiamremark{hypothesis}{Hypothesis}
\crefname{hypothesis}{Hypothesis}{Hypotheses}
\newsiamthm{claim}{Claim}
\usepackage{amsopn}

\usepackage{pgf}

\makeatletter 
\patchcmd\Gread@eps{\@inputcheck#1 }{\@inputcheck"#1"\relax}{}{}
\makeatother
\newcommand{\R}{\mathbb{R}}

\newcommand{\Z}{\bm{z}}

\newcommand{\fma}{{\tt fma}}
\newcommand{\fmas}{{\tt fmas}}
\usepackage{xpatch}
\makeatletter
\newcommand*{\addFileDependency}[1]{
  \typeout{(#1)}
  \@addtofilelist{#1}
  \IfFileExists{#1}{}{\typeout{No file #1.}}
}
\makeatother
\xpretocmd{\digraph}{\addFileDependency{#2.dot}}{}{}
\begin{document}
\newcommand\relatedversion{}
\renewcommand\relatedversion{\thanks{The full version of the paper can be accessed at \protect\url{arXiv.org}.}} 
%
\title{\Large Scheduled Jacobian Chaining\relatedversion}
    \author{Simon Märtens\thanks{corresponding author: \email{simon.maertens@rwth-aachen.de}.}
    \and Uwe Naumann\thanks{all: Informatik 12: Software and Tools for Computational Engineering, RWTH Aachen University, 52056 Aachen, Germany (\url{https://www.stce.rwth-aachen.de}).}}
\date{}
\maketitle
%
%
%
\input{content/00_abstract}
\input{content/01_introduction}
\input{content/02_scheduled_jacobian_chaining}
\input{content/03_results}
\input{content/05_implementation}
\input{content/06_conclusion}

\bibliographystyle{siamplain}
\bibliography{references}
\appendix
\input{content/99_appendix}

\end{document}

%% file: content/00_abstract.tex

\begin{abstract}
This paper addresses the efficient computation of Jacobian matrices for programs composed of sequential differentiable subprograms. By representing the overall Jacobian as a chain product of the Jacobians of these subprograms, we reduce the problem to optimizing the sequence of matrix multiplications, known as the {\sc Jacobian Matrix Chain Product} problem.

Solutions to this problem yield ``optimal bracketings'', which induce a precedence-constraint scheduling problem. We investigate the inherent parallelism in the solutions and develop a new dynamic programming algorithm as a heuristic that incorporates the scheduling. To assess its performance, we benchmark it against the global optimum, which is computed via a branch-and-bound algorithm.
\end{abstract}

%% file: content/01_introduction.tex

\section{Introduction.}
This contribution builds upon prior work on optimizing the cost of Jacobian matrix accumulations. The associated combinatorial challenge, known as the {\sc Optimal Jacobian Accumulation (OJA)} problem, has been proven to be NP-complete \cite{naumann_optimal_2008}. Elimination techniques like {\sc Vertex Elimination (VE)} \cite{forth_integer_2012, griewank_calculation_1991, naumann_optimal_2004, tadjouddine_vertex-ordering_2008}, {\sc Edge Elimination (EE)} \cite{naumann_optimal_2004}, and {\sc Face Elimination (FE)} \cite{naumann_optimal_2004, naumann_elimination_2023} present adjacent problem formulations that all have NP-hard algorithms for the general optimal solutions and heuristics that may produce near-optimal or even optimal results in certain cases \cite{griewank_accumulating_2003}. These algorithms all assume a general directed acyclic graph (DAG) structure of the program, such as the one depicted in Figure~\ref{fig:simple_example}. In this study, we are investigating a special case where the differentiable program can be implemented as a sequence of calls to subprograms. This leads to the well-known {\sc Jacobian Chain Product} problem \cite{godbole_efficient_1973}.

For the remainder of this section, we introduce said problem as well as the necessary preliminaries. Section~\ref{sec:scheduled_jacobian_chaining} presents our contribution, a dynamic programming (DP) formulation that generates near-optimal solutions to the new {\sc Scheduled Jacobian Chain Product} problem. We assess the quality of the solutions with a statistical analysis in Section~\ref{sec:results}. Finally, we close with a short introduction to our reference implementation in Section~\ref{sec:implementation} and an outlook and conclusion in Section~\ref{sec:conclusion}.
\input{figures/simple_example}

Let us consider a multivariate differentiable function (\emph{primal})
\begin{equation*}
    \bm{y} = F(\bm{x}) : \R^n \rightarrow \R^m
\end{equation*}
that can be decomposed into modules $F_i$ such that
\begin{equation*}
    F = F_q \circ F_{q-1} \circ \dots \circ F_2 \circ F_{1}
\end{equation*}
with
\begin{equation*}
    \bm{z}_i = F_i(\bm{z}_{i-1}) : \R^{n_i} \rightarrow \R^{m_i}
\end{equation*}
for $i = 1, \dots, q$ and $\bm{z}_0 = \bm{x}, \bm{z}_q = \bm{y}$. The chain rule of differentiation tells us that the Jacobian $F'$ can be written as
\begin{equation} \label{eqn:jcp}
    F' \equiv \frac{d F}{d \bm{x}}=F'_q \cdot F'_{q-1} \cdot \ldots \cdot F'_1 \in \R^{m \times n} \; .
\end{equation}
Without loss of generality, we are interested in the minimization of the computational cost in terms of scalar fused multiply-add (\fma) operations ($a \cdot b + c,$ $a,b,c \in \R$) performed during the evaluation of Equation~(\ref{eqn:jcp}). The number of scalar \fmas{} can be translated into more useful cost measures, such as run-time, through hardware-dependent benchmarks.
%
%
\begin{figure*}[b!]  
    \begin{equation*}
        \fma_{j,i} = \begin{cases}
            |E_j| \cdot \min(n_j,m_j) & j=i \\[5pt]
            \min_{i \leq k < j} \left( \fma_{j,k+1}+\fma_{k,i} + m_j \cdot m_k \cdot n_i \right) & j>i
        \end{cases}
    \end{equation*}
    \caption{DP recurrence relation for Dense Jacobian Chain Product Bracketing}
    \label{eqn:djcpb}
\end{figure*}
\subsection{Algorithmic differentiation (AD).}
Algorithmic differentiation (AD) \cite{griewank_evaluating_2008} offers two fundamental modes for accumulating the elemental Jacobians
\[
    F'_i = F'_i(\Z_{i-1}) \in \R^{m_i \times n_i}
\]
prior to the evaluation of the matrix chain product in Equation~(\ref{eqn:jcp}). Directional derivatives are computed in {\em vector tangent mode} as
\begin{equation} \label{eqn:vt}
    \dot{Z}_i = F_i'(\Z_{i-1}) \cdot \dot{Z}_{i-1} \in \R^{m_i \times \dot{n}_i}
\end{equation}
with $\dot{Z}_{i-1} \in \R^{n_i \times \dot{n}_i}$ can be used to efficiently accumulate the entire Jacobian by seeding $\dot{Z}_{i-1}$ with the identity $I_{n_i} \in \R^{n_i \times n_i}$. For given $\Z_{i-1} \in \R^{n_i}$ and $\dot{Z}_{i-1} \in \R^{n_i \times \dot{n}_i},$ the Jacobian-free evaluation of the {\em vector tangent mode} is denoted as
\begin{equation} \label{eqn:vt_mf}
	\dot{Z}_i = \dot{F}_i(\Z_{i-1}) \cdot \dot{Z}_{i-1} \in \R^{m_i \times \dot{n}_i} \; .
\end{equation}

Similarly, {\em vector adjoint mode}
\begin{equation} \label{eqn:va}
    \bar{Z}_{i-1} = \bar{Z}_i \cdot F_i'(\Z_{i-1}) \in \R^{\bar{m}_i \times n_i}
\end{equation}
with $\bar{Z}_i \in \R^{\bar{m}_i \times m_i}$ can be used to accumulate the entire Jacobian by seeding $\bar{Z}_i$ with the identity $I_{m_i} \in \R^{m_i \times m_i}$. For given $\Z_{i-1} \in \R^{n_i}$ and $\bar{Z}_i \in \R^{\bar{m}_i \times m_i},$ the Jacobian-free evaluation of the {\em vector adjoint mode} is denoted as
\begin{equation} \label{eqn:va_mf}
    \bar{Z}_{i-1} = \bar{Z}_i \cdot \bar{F}_i(\Z_{i-1}) \in \R^{\bar{m}_i \times n_i}\; .
\end{equation}
\setcounter{equation}{7} 
%
%
\subsection{Jacobian Chaining.}
The {\sc [Dense] Jacobian Chain Product Bracketing} problem asks for a bracketing of the right-hand side of Equation~(\ref{eqn:jcp}), which minimizes the number of \fma\ operations. This problem can be solved via dynamic programming \cite{bellman_dynamic_1957,godbole_efficient_1973}. The DP recurrence relation in Figure~\ref{eqn:djcpb} yields an optimal bracketing at a computational cost of $O(q^3)$. The variable $\fma_{j,i}$ in Figure~\ref{eqn:djcpb} encodes the computational cost of evaluating a sub-chain 
\[
    F'_{j,i} \equiv F'_{j} \cdot F'_{j-1} \cdot \ldots \cdot F'_{i+1} \cdot F'_{i}, \quad j> i,
\]
of Equation~(\ref{eqn:jcp}) which corresponds to the Jacobian of the primal function
\[
    F_{j,i} \equiv F_{j} \circ F_{j-1} \circ \dots \circ F_{i+1} \circ F_{i}, \quad j> i.
\]
Furthermore, $|E_j|$ can be interpreted as the number of edges in the DAGs $G_i=G_i(\Z_{i-1})=(V_i,E_i)$ for $i=1,\ldots,q$ of $F_i=F_i(\Z_{i-1})$. Using Equation~(\ref{eqn:vt_mf}), tangent propagation induces a cost of \mbox{$\dot{n}_i \cdot |E_i|$}. Similarly, using Equation~(\ref{eqn:va_mf}), adjoint propagation induces a cost of \mbox{$\bar{m}_i \cdot |E_i|$} as explained in detail in \cite{naumann_matrix-free_2024}.
\par
A new {\sc Matrix-Free Limited-Memory} variant of the problem was introduced in \cite{naumann_matrix-free_2024} that makes extensive use of the matrix-free vector modes in Equation~(\ref{eqn:vt_mf}) and Equation~(\ref{eqn:va_mf}). We go into more detail when we introduce our new DP formulation in Section~\ref{ssec:dp_with_scheduling}. For a full list of existing DP formulations, see Appendix~\ref{ssec:appendix_mflmdjcb_formulations}.
%
%
\subsection{Elimination Sequences.}
\label{ssec:elimination_sequences}
Several types of steps exist in the elimination sequences or bracketings in the case of Jacobian chains. There are accumulation (\texttt{ACC}) steps that calculate an elemental Jacobian using the corresponding tangent or adjoint models. Then there are different elimination (\texttt{ELI}) steps which use already calculated Jacobians to calculate the Jacobians of larger sub-chains. We use the following notation for these steps:
\begin{itemize}
    \item {
        \texttt{ACC TAN/ADJ (i-1 i)}\\
        Accumulation of elemental Jacobian $F_{i}' = \frac{d \bm{z}_i}{d \bm{z}_{i-1}}$ in either tangent (\texttt{TAN}) or adjoint (\texttt{ADJ}) mode by seeding the corresponding derivative models in Cartesian unit standard directions in $\R^{n_i}$ or $\R^{m_i}$ respectively, i.e.,
        \begin{equation}
            F_{i}' = \dot{F}_{i} \cdot I_{n_i} = I_{m_i} \cdot \bar{F}_{i} \,.
        \end{equation}
    }
    \item {
        \texttt{ELI TAN (i-1 k j)}\\
        Calculation of Jacobian $F_{j,i}' = \frac{d \bm{z}_j}{d \bm{z}_{i-1}}$ by seeding the tangent model of $F_{j, k+1}$ with the preaccumulated Jacobian $F_{k, i}' = \frac{d \bm{z}_k}{d \bm{z}_{i-1}}$, i.e.
        \begin{equation}
            F_{j,i}' = \dot{F}_{j, k+1} \cdot F_{k, i}' \,.
        \end{equation}
    }
    \item {
        \texttt{ELI ADJ (i-1 k j)}\\
        Calculation of Jacobian $F_{j,i}' = \frac{d \bm{z}_j}{d \bm{z}_{i-1}}$ by seeding the adjoint model of $F_{k, i}$ with the preaccumulated Jacobian $F_{j, k+1}' = \frac{d \bm{z}_j}{d \bm{z}_{k}}$, i.e.
        \begin{equation}
            F_{j,i}' = F_{j, k+1}' \cdot \bar{F}_{k, i} \,.
        \end{equation}
    }
    \item {
        \texttt{ELI MUL (i-1 k j)}\\
        Calculation of Jacobian $F_{j,i}' = \frac{d \bm{z}_j}{d \bm{z}_{i-1}}$ by multiplying the already accumulated Jacobians $F_{j, k+1}'$ and $F_{k, i}'$, i.e.
        \begin{equation}
            F_{j,i}' = F_{j, k+1}' \cdot F_{k, i}' \,.
        \end{equation}
    }
\end{itemize}
An example of a possible elimination sequence for a Jacobian chain of length six could look like this:
\begin{lstlisting}[basicstyle=\ttfamily\small]
    1: ACC TAN  (4 5)
    2: ELI TAN (4 5 6)
    3: ACC ADJ  (3 4)
    4: ELI ADJ (1 3 4)
    5: ELI MUL (1 4 6) 
    6: ELI ADJ (0 1 6)
\end{lstlisting}
Our DP solver optimized that sequence and calculated a predicted cost of \num{50977476} \fma. The corresponding bracketing can be written as
\begin{align*}
    F' &= F_{6,1}' \\ &= \left( \left( \dot{F}_6 \cdot \left(\dot{F}_5 \cdot I_{n_5}\right) \cdot \left(I_{m_4} \cdot \bar{F}_4 \cdot \bar{F}_3 \cdot \bar{F}_2 \right)\right)\right) \cdot \bar{F}_1 \;.
\end{align*}

%% file: figures/simple_example.tex
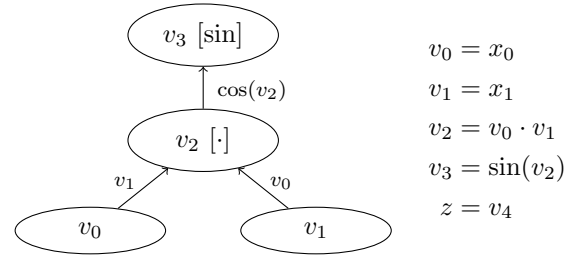
\begin{figure}[h]
    \centering
    \begin{tabular}{cc}
        \begin{minipage}[c]{.60\linewidth}
            \begin{tikzpicture}[scale=1, transform shape, ellipse]
                \begin{pgfscope}
                    \tikzstyle{every node}=[draw,ellipse,minimum width=2cm]
            	  \node (0) at (-1,0) {$v_0$};
            	  \node (1) at (2,0) {$v_1$};
            	  \node (2) at (0.5,1.2) {$v_2~[\cdot]$};
            	  \node (3) at (0.5,2.6) {$v_3~[\sin]$};
                \end{pgfscope}
                \begin{scope}[->]
                    \draw (0) -- (2) node[midway,left,xshift=4pt,yshift=2pt] {\footnotesize $v_1$};
                    \draw (1) -- (2) node[midway,right,xshift=-4pt,yshift=2pt] {\footnotesize $v_0$};
                    \draw (2) -- (3) node[midway,right,xshift=-5pt,yshift=-1pt] {\footnotesize $\cos(v_2)$};
              \end{scope}
            \end{tikzpicture}
        \end{minipage} &
        $\begin{aligned}
            v_0 &= x_0 \\
            v_1 &= x_1 \\
            v_2 &= v_0 \cdot v_1 \\
            v_3 &= \sin(v_2) \\
            z &= v_4
        \end{aligned}$ 
    \end{tabular}
    \caption{Example of a DAG labeled with the local derivatives and the corresponding primal single assignment code of the function $f: \R^2 \rightarrow \R$ with $\bm{x} \mapsto y \equiv f(\bm{x}) = \sin(x_0 \cdot x_1)$.} \label{fig:simple_example}
\end{figure}

%% file: content/02_scheduled_jacobian_chaining.tex

\section{\emph{Scheduled} Jacobian Chaining.}
\label{sec:scheduled_jacobian_chaining}
This section introduces our approach to generating elimination sequences (bracketings) that are optimized for parallel execution. There are multiple levels of parallelism that have to be considered:
\begin{enumerate}
    \item {
        \textbf{Primal code parallelism:} The primal code itself might already be parallelized with shared \cite{openmp_architecture_review_board_openmp_2021} or distributed memory \cite{message_passing_interface_forum_mpi_2023, naumann_framework_2008} management.
    }
    \item {
        \textbf{Preaccumulation and elimination:} A single step inside the elimination sequence is inherently parallel \cite{naumann_framework_2008}. AD vector modes can exploit loop-based parallelism or even SIMD instructions \cite{farrell_formal_1996} to accumulate the elemental Jacobians. 
    }
    \item {
        \textbf{Independence of elimination steps:} The elimination steps are partially independent of each other. For example, the preaccumulations can all run in parallel. The eliminations and multiplications depend on the availability of one or both local Jacobians of the sub-chains.
    }
\end{enumerate}
The first two types of parallelism are well-understood and handled by various AD software tools. This paper focuses on the last type.
%
%
\subsection{The Scheduling Problem.}
\label{ssec:scheduling_problem}
The bracketings induce parallelism based on tasks with precedence constraints. Using the notation introduced by \cite{graham_optimization_1979}, our scheduling problem could be classified as $\textrm{P}_m \, | \, \textrm{prec} \, | \, C_{\textrm{max}}$, i.e. we minimize the overall completion time (\emph{makespan}) on $m$ identical machines while respecting precedence constraints. The precedence constraints for the bracketings are somewhat special because a task can only be the predecessor of exactly one other task. Any intermediate Jacobian is only accumulated and used once. Therefore, the constraints can be visualized as a DAG where all nodes (tasks) have an out-degree of 1. The latter is called an \emph{in-tree}, which gives us the final classification $\textrm{P}_m \, | \, \textrm{intree} \, | \, C_{\textrm{max}}$. Figure~\ref{fig:task_dependencies} shows the in-tree dependencies for the example at the end of Section~\ref{ssec:elimination_sequences}.
\begin{figure}[h]
    \centering
    \digraph[scale=0.5]{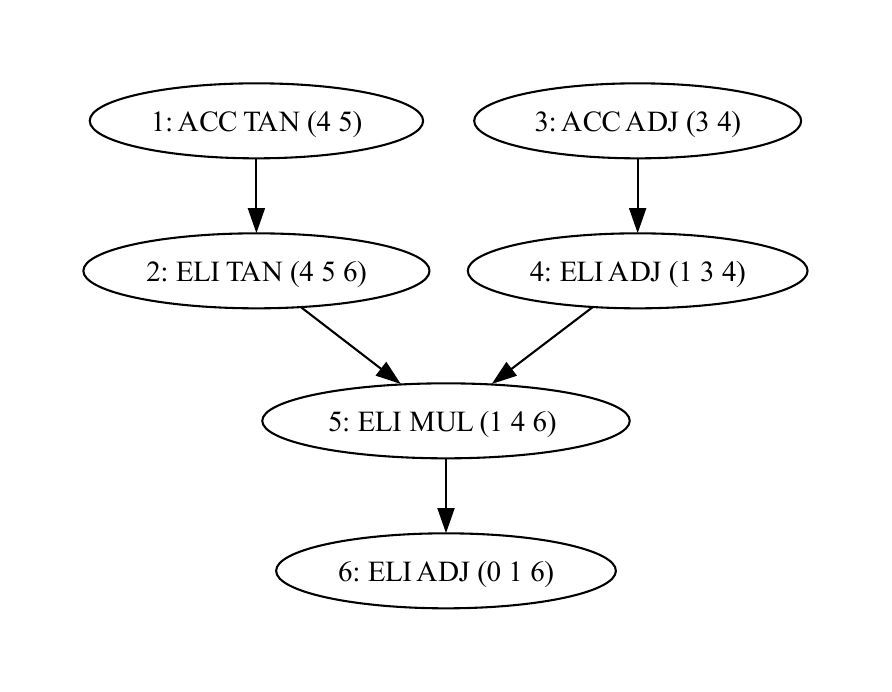}{
        1 [label="1: ACC TAN  (4 5)"];
        2 [label="2: ELI TAN (4 5 6)"];
        3 [label="3: ACC ADJ  (3 4)"];
        4 [label="4: ELI ADJ (1 3 4)"];
        5 [label="5: ELI MUL (1 4 6) "];
        6 [label="6: ELI ADJ (0 1 6)"];
        1 -> 2;
        2 -> 5;
        3 -> 4;
        4 -> 5;
        5 -> 6;
    }
    \caption{Task dependencies for example elimination sequence at the end of Section~\ref{ssec:elimination_sequences}.}
    \label{fig:task_dependencies}
\end{figure}

Each one of the $n$ elimination steps is a job $j \in J = [1,n]$ and has a given runtime $p_j$. The $m$ identical machines are assigned an identifier $i \in M = [1, m]$. The goal of our scheduling algorithms is to assign a machine $\sigma(j)$ as well as a start time $S_j$ to each job. Since we are only considering non-preemptive scheduling (the jobs cannot be canceled), the completion time is $C_j = S_j + p_j$. The \emph{makespan} $C_{\textrm{max}} = \max_j C_j$ is our optimization objective subject to the precedence constraints $\hat{E}$:
$$
    (j', j) \in \hat{E} \,\, \textrm{with} \,\, j',j \in J \,\Rightarrow\, j' \prec j \,\Leftrightarrow\, S_j \geq C_{j'} \; .
$$
The time when a machine becomes available is denoted by $\alpha_i$. When nothing is scheduled $\left(\sigma(j) = \varnothing, \, \forall j \in J \right)$ we assume that all machines are idling, i.e., $\alpha_i = 0$. Their availability changes during the scheduling according to
$$
    \alpha_i = \max_{j, \sigma(j) = i}{C_j} \; .
$$
The problem is NP-hard, as shown in \cite{baskiyar_scheduling_2001, karp_reducibility_1972}. Restricted formulations with equal processing times $p_j = 1$ are computationally tractable \cite{hu_parallel_1961}, but rarely applicable to our problem. All Jacobians in the chain would need to be of equal size, and the AD model evaluations would need to cost the same as a multiplication between two Jacobians. We will still use the algorithms as a heuristic. Further information on the general scheduling problem is provided in \cite{graham_bounds_1966, graham_bounds_1971}.
\par
There are more complex formulations that consider communication delays between jobs \cite{davies_scheduling_2020}. In the context of Jacobian chain bracketings, these delays might be significant depending on the size of the elemental Jacobians, which have to be stored or transferred between machines. This aspect is the subject of ongoing research.
%
%
\subsection{DP with Scheduling.}
\label{ssec:dp_with_scheduling}
\begin{figure*}[ht]
    \centering
    \begin{equation*}
        \fma_{j,i}^{(t)} = \begin{cases}
            |E_j| \cdot \begin{cases}
                n_j & |E_j| > \overline{M} \\
                \min(n_j,m_j) & \text{otherwise}
            \end{cases} & j=i \\[15pt]
            \min_{i \leq k < j} \left( \min \left(
            \begin{aligned}
                &m_j \cdot m_k \cdot n_i + \min \left(
                \begin{aligned}
                    &\fma_{j,k+1}^{(t)}+\fma_{k,i}^{(t)} \\
                    &\min_{1 \leq t^{*} < t} \left( \max\left(\fma_{j,k+1}^{(t^{*})}, \fma_{k,i}^{(t - t^{*})}\right) \right)
                \end{aligned} \right) \\
                &\fma_{j,k+1}^{(t)} + m_j \cdot \sum_{\nu=i}^k |E_\nu| \quad \text{if}~\sum_{\nu=i}^k |E_\nu| \leq \overline{M} \\
                &\fma_{k,i}^{(t)} + n_i \cdot \sum_{\nu=k+1}^j |E_\nu|
            \end{aligned}
            \right) \right) & j>i \; .
        \end{cases}
    \end{equation*}
    \caption{DP relation for Scheduled Limited-Memory Matrix-Free Dense Jacobian Chain Product Bracketing.}
    \label{eqn:scheduled_chaining}
\end{figure*}
Because the underlying scheduling problem is NP-hard, the {\sc Scheduled Jacobian Chain Product} problem is also NP-hard. Therefore, no tractable DP formulation guarantees an optimal solution. Our goal is to introduce scheduling into the DP formulation to produce a heuristic that produces near-optimal results on average.
\par
We augment each sub-problem $\fma_{j, i}$ in the previous DP formulation (Appendix~\ref{ssec:appendix_mflmdjcb_formulations}) with the number of machines/threads $t \leq m$ that is available to solve it: $\fma_{j, i}^{(t)}$. This step allows us to modify the DP formulation as follows:
\par
\begin{flushleft}
    \textbf{Accumulation}
\end{flushleft}
\[
    |E_j| \cdot \begin{cases}
        n_j & |E_j| > \overline{M} \\
        \min(n_j,m_j) & \text{otherwise}
    \end{cases}
\]
The accumulation of elemental Jacobians does not depend on any other elimination step. We need to consider the cost of accumulating the elemental Jacobian in tangent and adjoint mode and choose the one that is cheaper. A memory limit $\overline{M}$ was introduced in \cite{naumann_matrix-free_2024} that decides whether the adjoint mode is feasible due to the fact that a tape needs to be stored in persistent memory. The only consideration for us here is the definition $\overline{M}$ in the parallel setting. In distributed-memory frameworks, we can define it as the available memory per machine/thread $\overline{M}_{\textrm{distributed}} = \frac{\overline{M}_{\textrm{total}}}{m}$. If the total memory is shared, however, the memory limit would actually increase with the number of threads that are allocated for the accumulation $\overline{M}_{\textrm{shared}} = t \cdot \frac{\overline{M}_{\textrm{total}}}{m}$.
\par
\begin{flushleft}
    \textbf{Elimination}
\end{flushleft}
\[
    \fma_{j,k+1}^{(t)} + m_j \cdot \sum_{\nu=i}^k |E_\nu| \quad \text{if}~\sum_{\nu=i}^k |E_\nu| \leq \overline{M}
\]
\[
    \fma_{k,i}^{(t)} + n_i \cdot \sum_{\nu=k+1}^j |E_\nu|
\]
The matrix-free adjoint and tangent eliminations for the sub-problem $\fma_{j, i}$ depend on a preaccumulated Jacobian $\fma_{j,k+1}$ or $\fma_{k,i}$, respectively. The cost of evaluating the rest of the Jacobian chain in adjoint or tangent mode is added to the cost of the preaccumulated Jacobian. The eliminations are evaluated serially due to the precedence constraints. The number of available threads/machines is completely propagated to the sub-problems $\fma_{j,k+1}^{(t)}$ or $\fma_{k,i}^{(t)}$. The same consideration about the memory limit for the adjoint mode applies.
\par
\begin{flushleft}
    \textbf{Multiplication}
\end{flushleft}
\[
    m_j \cdot m_k \cdot n_i + \min \left(
        \begin{aligned}
            &\fma_{j,k+1}^{(t)}+\fma_{k,i}^{(t)} \\
            &\min_{1 \leq t^{*} < t} \left( \max\left(\fma_{j,k+1}^{(t^{*})}, \fma_{k,i}^{(t - t^{*})}\right) \right)
        \end{aligned} \right)
\]
The multiplication of two Jacobians is the only task that allows us to parallelize because it depends on two independent sub-problems. We divide the number of available threads between the two 
sub-problems, $\fma_{j,k+1}^{(t^{*})}$ and $\fma_{k, i}^{(t - t^{*})}$ with $1 \leq t^{*} < t$. The sub-problems are now considered parallel, and we only count the maximum cost between these two. There is also the option to propagate the full amount of available threads to both sub-problems and run them in serial. There are cases where that is the best option, even if more than one thread is available. For example, when the two sub-problems each contain multiple multiplication steps that can all run in parallel. If that's the case, it can be better to run the sub-problems one after the other with the maximum number of threads each.
\par
The final DP formulation is shown in Figure~\ref{eqn:scheduled_chaining}. Formulations without memory limits and without matrix-free eliminations are listed in Appendix~\ref{ssec:appendix_scheduled_formulations}.
\par
The recurrence relation in Figure~\ref{eqn:scheduled_chaining} does not actually create a schedule; it just assigns a number of available threads to each sub-problem. To create the final schedule, we need to create a thread pool that we then partition among the tasks. Let us assume the general case of a chain of length $q$ and $m$ available machines/threads. For each triplet $(j, i, t)$ with $1 \leq i \leq j \leq q$ and $1 \leq t \leq m$ the DP table stores the optimal $\fma$ as well as additional information, most importantly the values for $k$ and $t^{*}$ (split positions of the chain and the number of available threads). Using this information, one can backtrack the best sequence. We start at the node for $\fma_{q,1}^{(t)}$ giving it the entire available machine pool $[i_l, i_u] = M = [1, t] $ with $i_l$ and $i_u$ the lower and upper bounds for the machine identifiers. From there, we pass the entire machine pool to the sub-problems, except if we encounter a multiplication that divides the available threads among its sub-problems. We split the machine pool accordingly; $\fma_{j,k+1}^{(t^{*})}$ receives $[i_l, i_l + t^{*} - 1]$ and $\fma_{k, i}^{(t - t^{*})}$ receives $[i_u - t^{*} + 1, i_u]$. The notation for the steps in the elimination sequences is extended with the machine pool. If only one machine is available, then we only note the machine ID. Finally, if we have the complete sequence, then we choose the machines with the lowest identifier in each task's machine pool.

Applying the scheduled DP to the Jacobian chain from the example at the end of Section~\ref{ssec:elimination_sequences} with 3 available machines, we get the following improved elimination sequence with machine pool assignments:
\begin{lstlisting}[basicstyle=\ttfamily\small]
    1: ACC TAN  (4 5)  [1]
    2: ELI TAN (4 5 6) [1]
    3: ACC ADJ  (3 4)  [2]
    4: ELI ADJ (2 3 4) [2]
    5: ACC TAN  (1 2)  [3]
    6: ELI MUL (1 2 4) [2,3]
    7: ELI ADJ (0 1 4) [2,3]
    8: ELI MUL (0 4 6) [1,3]
\end{lstlisting}
Our scheduled DP solver optimized that sequence and predicted an evaluation cost of \num{29797092}, which is about $60\%$ of the sequential cost. Figure~\ref{fig:task_dependencies_scheduled} shows the task dependencies of the scheduled elimination sequence.
\begin{figure}[h]
    \hspace{-37pt}
    \digraph[scale=0.49]{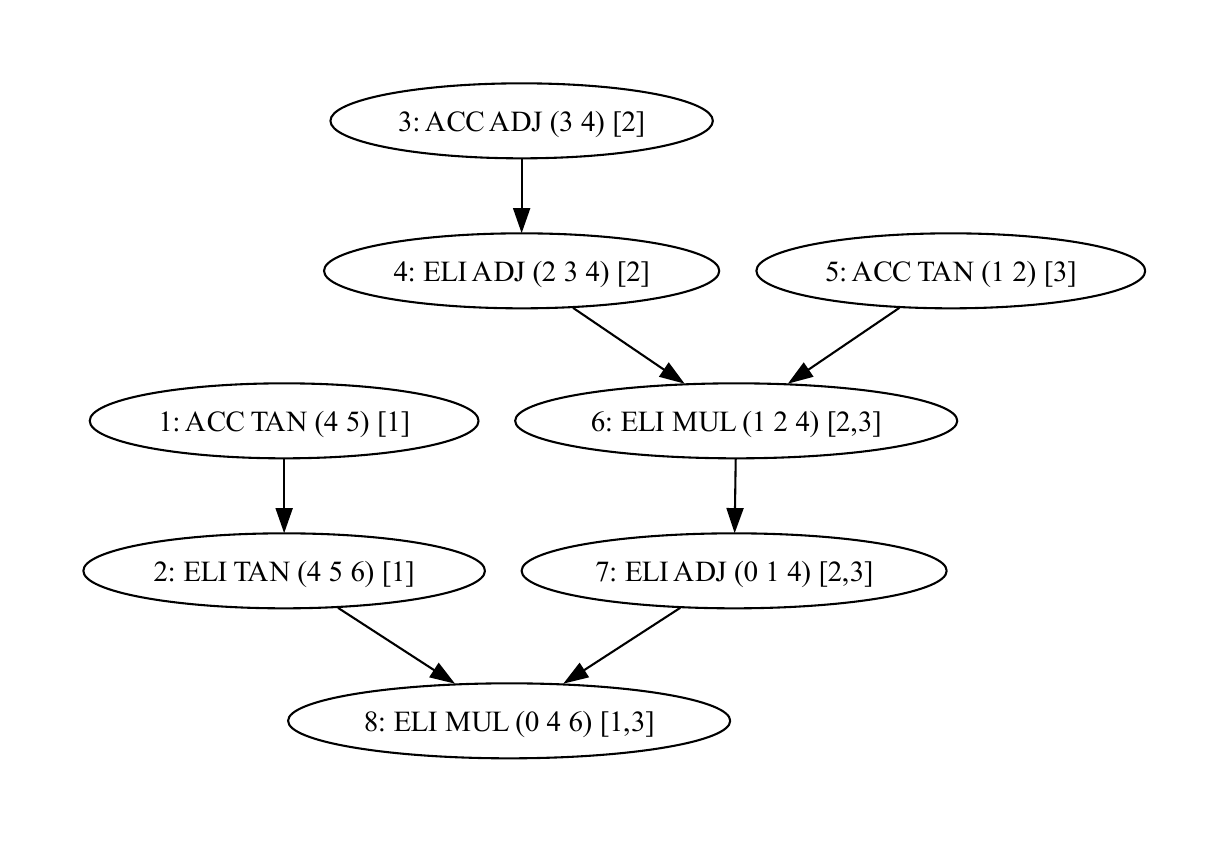}{
        1 [label="1: ACC TAN (4 5) [1]"];
        2 [label="2: ELI TAN (4 5 6) [1]"];
        3 [label="3: ACC ADJ (3 4) [2]"];
        4 [label="4: ELI ADJ (2 3 4) [2]"];
        5 [label="5: ACC TAN (1 2) [3]"];
        6 [label="6: ELI MUL (1 2 4) [2,3]"];
        7 [label="7: ELI ADJ (0 1 4) [2,3]"];
        8 [label="8: ELI MUL (0 4 6) [1,3]"];
        1 -> 2
        2 -> 8
        3 -> 4
        4 -> 6
        5 -> 6
        6 -> 7
        7 -> 8
    }
    \caption{Task dependencies for example elimination sequence with 3 available machines.}
    \label{fig:task_dependencies_scheduled}
\end{figure}
%
%
\subsection{Optimality.}
\label{ssec:optimality}
Since the formulation in Figure~\ref{eqn:scheduled_chaining} does not guarantee optimality, the classical DP considerations about \emph{optimal substructure} and \emph{overlapping sub-problems} do not apply for arbitrary $t$. For $t=1$, the formulations collapse back into the previous formulations in Appendix~\ref{ssec:appendix_mflmdjcb_formulations}. For $t=\infty$ (unlimited machines/threads), the multiplication term can be reduced to
\begin{equation}
    \fma_{j,i}^{(\infty)} = m_j \cdot m_k \cdot n_i + \max\left(\fma_{j,k+1}^{(\infty)}, \fma_{k,i}^{(\infty)}\right)
\end{equation}
because it does not matter how we split the threads. The sub-problems will always have unlimited threads available.
\par 
\begin{theorem}
    Figure~\ref{eqn:scheduled_chaining} with $t=\infty$ solves the {\sc Scheduled Limited-Memory Matrix-Free Dense Jacobian Chain Product Bracketing} problem.
\end{theorem}
\begin{proof}
    {\sc Scheduled Dense Jacobian Chain Product Bracketing} inherits the \emph{overlapping sub-problems} property from {\sc Dense Jacobian Chain Product Bracketing} since the sub-problems do not change. Replacing the serial execution of the sub-problems ($\fma_{j,k+1} + \fma_{k,i}$) with its parallel counterpart $\max\left(\fma_{j,k+1}, \fma_{k,i}\right)$ may only reduce the cost of $\fma_{j,i}$. Following the proof in \cite{naumann_matrix-free_2024}, the \emph{optimal substructure} is also preserved and will be inherited by the {\sc Limited-Memory Matrix-Free} variants.
\end{proof}
In reality, the number of available machines is limited. However, the length of the Jacobian chain and the number of accumulations in the elimination sequence are natural limiters for the number of machines we can use. Therefore, if we use as many machines as there are accumulations, the scheduled DP generates an optimal sequence.
%
%
\subsection{Correctness and Determinism.}
\label{ssec:correctness}
\begin{theorem}
    The elimination sequences produced by Figure~\ref{eqn:scheduled_chaining} deterministically calculate the correct Jacobian matrix.
\end{theorem}
\begin{proof}
    The search space of Figure~\ref{eqn:scheduled_chaining} is the same as the search space of its serial counterpart, extended by a subspace of possible thread partitions. In other words, for each possible elimination sequence, a certain number of possible schedules is considered. Thus, all elimination sequences produced by Figure~\ref{eqn:scheduled_chaining} can run in serial and will produce the same correct Jacobian (in infinite precision).

    As explained above, Figure~\ref{eqn:scheduled_chaining} does not actually create a schedule; it just assigns a machine pool to each task. Each pool will contain at least one machine. Given that we adhere to the precedence constraints in the sequence, no task will start before all necessary sub-problems are done. Thus, deadlocks are impossible, and the program will terminate.
\end{proof}

%% file: content/03_results.tex

\section{Statistical Assessment.}
\label{sec:results}
To assess the quality of the elimination sequences produced by the presented algorithm, we also implemented a nested Branch \& Bound (B\&B) algorithm that evaluates all possible schedules for all possible elimination sequences for a given Jacobian chain. It results in a global optimum, which we can compare with the DP solution.
\par
We generate $10^4$ random Jacobian chains of constant length $q$. The sizes $m_i, n_i$ are sampled from a uniform distribution $m_i, n_i \sim \mathcal{U}(5,50)$ and the size of the DAGs are sampled from $E \sim \mathcal{U}(10^3,10^4)$. For each chain the DP and B\&B solvers generate sequences with makespans $\fma_{DP}^{(m)}$ and $\fma_{opt}^{(m)}$ ($1 \leq m \leq q$). The ratio
\begin{equation}
    \mathcal{C}^{(m)} = \frac{\fma_{opt}^{(m)}}{\fma_{DP}^{(m)}} \leq 1
\end{equation}
is the measure we are analyzing. This is the ratio between the cost of the optimal solution and the cost of the sequence the DP solver generated. The ratio can also be interpreted as a percentage. For example, $\mathcal{C}^{(m)} = 0.9$ means that the optimal solution needs $10\%$ less \fmas{} than the sequence the DP solver generated. The theoretical worst-case would be
\begin{equation}
    \mathcal{C}_{WC}^{(m)} = \frac{1}{\min(m, \tilde{m})}
\end{equation}
with $\tilde{m}$ being the number of useful machines, i.e., the number of machines used by the solution with unlimited parallelism. A proof for this can be found in Appendix~\ref{ssec:appendix_worstcase_analysis}.
\begin{figure}[h]
    \centering
    \input{figures/results_length_6}
    \caption{Ratio between $\fma_{opt}^{(m)}$ and $\fma_{DP}^{(m)}$ calculated for $10^4$ random chains of length $q=6$.}
    \label{fig:results_length_6}
\end{figure}
Figure~\ref{fig:results_length_6} shows the results for chains of length $q=6$. The boxes mark the lower (first) and upper (third) quartiles, with the median as the line inside the box. The whiskers reach from the 2nd percentile to the 98th percentile, and the crosses are outliers. Further benchmarks are listed in Appendix~\ref{ssec:appendix_stats}. Additionally, average and worst performances are listed in Table~\ref{tab:average_dp_results} and \ref{tab:worst_dp_results}.
\begin{table}[h]
    \centering
    \begin{tabular}{|c||c|c|c|c|c|c|c|} \hline
        $q \backslash m$ & $2$    & $3$    & $4$    & $5$    & $6$    & $7$    & $8$   \\ \hline\hline
                     $2$ & $1$    &        &        &        &        &        &       \\ \hline
                     $3$ & $.987$ & $1$    &        &        &        &        &       \\ \hline
                     $4$ & $.965$ & $.993$ & $1$    &        &        &        &       \\ \hline
                     $5$ & $.950$ & $.980$ & $.998$ & $1$    &        &        &       \\ \hline
                     $6$ & $.938$ & $.966$ & $.992$ & $.999$ & $1$    &        &       \\ \hline
                     $7$ & $.930$ & $.953$ & $.992$ & $.999$ & $1.00$ & $1$    &       \\ \hline
                     $8$ & $.923$ & $.949$ & $.991$ & $.998$ & $1.00$ & $1.00$ & $1$   \\ \hline
    \end{tabular}
    \caption{Average $\mathcal{C}^{(m)}$ for $2 \leq q \leq 8$ and $2 \leq m \leq 8$.}
    \label{tab:average_dp_results}
\end{table}
\begin{table}[h]
    \centering
    \begin{tabular}{|c||c|c|c|c|c|c|c|} \hline
         $q \backslash m$ & $2$    & $3$    & $4$    & $5$    & $6$    & $7$    & $8$   \\ \hline\hline
                      $2$ & $1$    &        &        &        &        &        &       \\ \hline
                      $3$ & $.733$ & $1$    &        &        &        &        &       \\ \hline
                      $4$ & $.717$ & $.733$ & $1$    &        &        &        &       \\ \hline
                      $5$ & $.705$ & $.688$ & $.770$ & $1$    &        &        &       \\ \hline
                      $6$ & $.713$ & $.705$ & $.735$ & $.825$ & $1$    &        &       \\ \hline
                      $7$ & $.724$ & $.720$ & $.758$ & $.804$ & $.899$ & $1$    &       \\ \hline
                      $8$ & $.731$ & $.728$ & $.761$ & $.802$ & $.865$ & $.931$ & $1$   \\ \hline
    \end{tabular}
    \caption{Minimal $\mathcal{C}^{(m)}$ for $2 \leq q \leq 8$ and $2 \leq m \leq 8$.}
    \label{tab:worst_dp_results}
\end{table}

From the tables, we can clearly see a trend in which the solutions from the DP solver worsen on average with increasing chain length and improve with increasing machines. The worst performances, i.e., the minimal $\mathcal{C}^{(m)}$ we encountered, stagnate around the same values with a small trend to improvements with more machines. We can, however, construct a synthetic sequence that tends towards the worst-case $\mathcal{C}_{WC}^{(m)}$ if we let certain Jacobian and DAG sizes tend towards infinity. These extreme cases are less likely to occur with longer sequences and more machines.
\par
Finally, Table~\ref{tab:optimal_dp_precentages} shows the percentage of how many times the DP solver was able to find the global optimum, i.e., $\mathcal{C}^{(m)} = 1$. We stop with the benchmarks at $q=8$ because, for some larger chains, the B\&B algorithm was not feasible any longer with run times up to several hours.
\begin{table}[H]
    \centering
    \begin{tabular}{|c||c|c|c|c|c|c|c|} \hline
         $q \backslash m$ & $2$    & $3$    & $4$    & $5$    & $6$    & $7$    & $8$   \\ \hline\hline
                      $2$ & $100$  &        &        &        &        &        &       \\ \hline
                      $3$ & $83.3$ & $100$  &        &        &        &        &       \\ \hline
                      $4$ & $59.6$ & $89.8$ & $100$  &        &        &        &       \\ \hline
                      $5$ & $43.5$ & $74.3$ & $95.6$ & $100$  &        &        &       \\ \hline
                      $6$ & $32.9$ & $56.7$ & $86.6$ & $98.3$ & $100$  &        &       \\ \hline
                      $7$ & $25.1$ & $43.7$ & $84.1$ & $96.6$ & $99.5$ & $100$  &       \\ \hline
                      $8$ & $19.4$ & $38.5$ & $80.7$ & $93.6$ & $98.3$ & $99.7$ & $100$ \\ \hline
    \end{tabular}
    \caption{Percentage of how many times the DP solver was able to find the global optimum, i.e., $\mathcal{C}^{(m)} = 1$.}
    \label{tab:optimal_dp_precentages}
\end{table}

%% file: content/05_implementation.tex

\section{Implementation.}
\label{sec:implementation}
Our reference implementation of the DP and B\&B algorithms can be found at:
\begin{center}
    \url{https://doi.org/10.5281/zenodo.15373160}
\end{center}
\par
After following the steps in the \texttt{README.md} file there should be two executables: \texttt{jcdp} and \texttt{jcdp\_batch}. The former generates one Jacobian chain, executes the DP and B\&B solvers, and outputs all obtained solutions and several other statistics. The latter generates several random chains, and it writes the makespans of the different solutions into a file. It is used to generate the dataset for the statistical analysis. Both programs expect a configuration file, of which there are several in \texttt{additionals/configs/}. The config file that produces the example sequence at the end of Section~\ref{ssec:elimination_sequences} is shown in Figure~\ref{fig:config_file}.
\begin{figure}[h]
    \centering
    \begin{lstlisting}
        length 10
        size_range 5 500
        dag_size_range 1000 100000
        available_threads 1
        available_memory 0
        matrix_free 1
        time_to_solve 5
        seed 2165743199
    \end{lstlisting}
    \caption{Configuration file for the example at the end of Section~\ref{ssec:elimination_sequences}.}
    \label{fig:config_file}
\end{figure}
The plots in this paper, as well as the data in the tables, can be generated by the Python script \texttt{additionals/scripts/generate\_plots.py}. The datasets and config files that are used to generate the results in this paper are all collected in \texttt{additionals/acda25}. Docker containers are supplied for convenience. Refer to the \texttt{README.md} for further instructions. 

%% file: content/06_conclusion.tex

\section{Conclusion and Future Work.}
\label{sec:conclusion}
%


%
This paper addresses the inherent parallelism induced by Jacobian chain products for programs composed of sequential differentiable subprograms. Specifically, we analyze how different bracketings of the chain product correspond to different precedence-constrained scheduling problems. We introduce a new DP algorithm that incorporates scheduling considerations directly into the bracketing optimization process. Despite the NP-hardness of the underlying scheduling problem, our DP approach serves as an effective heuristic. We validate its performance through statistical analysis, comparing the makespan of schedules produced by our algorithm against the global optimum obtained via a B\&B method.

The results demonstrate that our algorithm consistently generates near-optimal schedules, with the average performance gap decreasing as the number of available machines increases. Even in cases with limited machines, the schedules produced are within a practical range of the optimal makespan. This indicates that our method effectively exploits available parallelism without incurring significant computational overhead or deviating substantially from the optimal solution.

\paragraph{Case Studies.}

Given that the produced results are purely statistical, the logical next step would be to apply the scheduled Jacobian chaining to some real-world examples. We envision the following:

Given a tape-based AD tool that can interpret its tape in tangent (forward) and adjoint (reverse) mode, we split the recorded tape into chunks. These chunks, combined with the forward and reverse interpretation, act as the tangent and adjoint models of the elemental Jacobians. We can invoke our DP solver at runtime and interpret the tape in parallel via a task-based shared-memory framework like OpenMP. The main challenge here is the partitioning of the tape. For example, we may want to split the tape at positions where the number of nodes at the interfaces that connect the new DAGs is minimal to reduce the sizes of the elemental Jacobians.

\paragraph{Communication Costs.}

In parallel computing environments, especially in distributed-memory systems, communication delays between machines can impact overall performance. Expanding to scheduling algorithms that account for communication overhead, particularly when transferring large intermediate Jacobian matrices between machines, seems desirable.

\paragraph{Sparse Jacobian Computations.}

Many practical problems involve sparse Jacobian matrices. Incorporating sparsity into our DP and scheduling framework could lead to substantial reductions in computational cost and memory usage. We can build on existing work \cite{griewank_accumulating_2003} and incorporate Jacobian compression \cite{gebremedhin_what_2005, griewank_evaluating_2008} into the elimination steps, potentially altering the evaluation costs quite substantially. However, the scheduling aspect introduced in this paper is independent of the sparsity of the Jacobians, because the precedence constraints are unaffected by sparsity. Only the costs of the individual elimination steps might change. Thus, the statistical results presented should hold regardless of whether the considered Jacobians are dense or sparse.

%% file: content/99_appendix.tex

\newpage
\section{Appendix}
\label{sec:appendix}
%
%
\subsection{Worst-case Analysis.}
\label{ssec:appendix_worstcase_analysis}
We can derive a lower bound for the ratio between the optimal solution of the scheduled bracketing and the solution that our DP algorithm can produce. The following four observations are needed:
\begin{enumerate}
    \item {
        The space of all possible elimination sequences with $m-1$ machines is fully contained within the solution space with $m$ machines because not all $m$ machines must be used for the solution. This means that the DP solution for $m$ machines is never larger than the solution for $m-1$ machines and, therefore, also never larger than the serial solution.
        \begin{equation}
            \fma_{DP}^{(m)} \leq \fma_{DP}^{(m-1)} \leq \dots \leq \fma_{DP}^{(1)}
        \end{equation}
    }
    \item {
        A perfect scaling in the optimal solution with $m$ machines cannot exist because the last operation, whether it is a multiplication of two Jacobians or an evaluation of a tangent or adjoint model, is dependent on all previous operations.
        \begin{equation}
            \fma_{opt}^{(m)} > \frac{1}{m} \fma_{opt}^{(1)}
        \end{equation}
    }
    \item {
        The serial DP solution ($m=1$) is optimal \cite{naumann_matrix-free_2024}.
        \begin{equation}
             \fma_{opt}^{(1)} = \fma_{DP}^{(1)}
        \end{equation}
    }
    \item {
        The optimal solution with $m$ machines might not be able to use all of them due to the precedence constraints. If that is the case, the optimal solution will not improve with more machines. The number of \emph{useful} machines may, therefore, be limited by $m^*$ smaller than or equal to the length of the Jacobian chain.
        \begin{equation}
             \fma_{opt}^{(m)} \geq \fma_{opt}^{(m^*)} \qquad 1 < m^* \leq q
        \end{equation}
    }
\end{enumerate}
Equipped with these observations, we can derive the following lower bounds:
\paragraph{Case $m \leq m^*$.}
\begin{equation*}
    \mathcal{C}^{(m)} = \frac{\fma_{opt}^{(m)}}{\fma_{DP}^{(m)}} \geq \frac{\fma_{opt}^{(m)}}{\fma_{DP}^{(1)}} > \frac{\frac{1}{m} \fma_{opt}^{(1)}}{\fma_{DP}^{(1)}} = \frac{1}{m}
\end{equation*}
\paragraph{Case $m > m^*$.}
\begin{equation*}
    \mathcal{C}^{(m)} = \frac{\fma_{opt}^{(m)}}{\fma_{DP}^{(m)}} \geq \frac{\fma_{opt}^{(m)}}{\fma_{DP}^{(1)}} \geq \frac{\fma_{opt}^{(m^*)}}{\fma_{DP}^{(1)}} > \frac{\frac{1}{m^*} \fma_{opt}^{(1)}}{\fma_{DP}^{(1)}} = \frac{1}{m^*}
\end{equation*}
The final worst-case lower-bound is therefore
\begin{equation}
    \mathcal{C}^{(m)} > \mathcal{C}_{WC}^{(m)} = \frac{1}{\min(m, m^*)} \qquad 1 < m^* \leq q.
\end{equation}
%
%
\onecolumn
\subsection{\emph{Matrix-Free} \emph{Limited-Memory} Dense Jacobian Chain Bracketing Formulations.}
\label{ssec:appendix_mflmdjcb_formulations}
\begin{flushleft}
    \vfill
    \textbf{Dense Jacobian Chain Product Bracketing}
\end{flushleft}
\begin{equation} \label{eqn:djcpb_2}
    \fma_{j,i} = \begin{cases}
        |E_j| \cdot \min(n_j,m_j) & j=i \\[5pt]
        \min_{i \leq k < j} \left( \fma_{j,k+1}+\fma_{k,i} + m_j \cdot m_k \cdot n_i \right) & j>i \; .
    \end{cases}
\end{equation}
\begin{flushleft}
    \vfill
    \textbf{\emph{Matrix-Free} Dense Jacobian Chain Product Bracketing}
\end{flushleft}
\begin{equation} \label{eqn:mfdjcpb}
    \fma_{j,i} = \begin{cases}
        |E_j| \cdot \min(n_j,m_j) & j=i \\[10pt]
        \min_{i \leq k < j} \left( \min \left(
        \begin{aligned}
            &\fma_{j,k+1}+\fma_{k,i} + m_j \cdot m_k \cdot n_i \\
            &\fma_{j,k+1} + m_j \cdot \sum_{\nu=i}^k |E_\nu| \\
            &\fma_{k,i} + n_i \cdot \sum_{\nu=k+1}^j |E_\nu|
        \end{aligned}
        \right) \right) & j>i \; .
    \end{cases}
\end{equation}
\begin{flushleft}
    \vfill
    \textbf{\emph{Limited-Memory} Matrix-Free Dense Jacobian Chain Product Bracketing}
\end{flushleft}
\begin{equation} \label{eqn:lmmfdjcpb}
    \fma_{j,i} = \begin{cases}
        |E_j| \cdot \begin{cases}
            n_j & |E_j| > \overline{M} \\
            \min(n_j,m_j) & \text{otherwise}
        \end{cases} & j=i \\[15pt]
        \min_{i \leq k < j} \left( \min \left(
        \begin{aligned}
            &\fma_{j,k+1}+\fma_{k,i} + m_j \cdot m_k \cdot n_i \\
            &\fma_{j,k+1} + m_j \cdot \sum_{\nu=i}^k |E_\nu| \quad \text{if}~\sum_{\nu=i}^k |E_\nu| \leq \overline{M} \\
            &\fma_{k,i} + n_i \cdot \sum_{\nu=k+1}^j |E_\nu|
        \end{aligned}
        \right) \right) & j>i \; .
    \end{cases}
\end{equation}
\begin{flushleft}
    \vfill
\end{flushleft}
%
%
\subsection{\emph{Scheduled} Jacobian Chain Bracketing Formulations.}
\label{ssec:appendix_scheduled_formulations}
\begin{flushleft}
    \vfill
    \textbf{\emph{Scheduled} Dense Jacobian Chain Product Bracketing}
\end{flushleft}
\begin{equation} \label{eqn:djcpb_scheduled}
    \fma_{j,i}^{(t)} = \begin{cases}
        |E_j| \cdot \min(n_j,m_j) & j=i \\[5pt]
        \min_{i \leq k < j} \left( m_j \cdot m_k \cdot n_i + \min \left(
        \begin{aligned}
            &\fma_{j,k+1}^{(t)}+\fma_{k,i}^{(t)} \\
            &\min_{1 \leq t^{*} < t} \left( \max\left(\fma_{j,k+1}^{(t^{*})}, \fma_{k,i}^{(t - t^{*})}\right) \right)
        \end{aligned}
        \right) \right) & j>i \; .
    \end{cases}
\end{equation}
\begin{flushleft}
    \vfill
    \textbf{\emph{Scheduled} Matrix-Free Dense Jacobian Chain Product Bracketing}
\end{flushleft}
\begin{equation} \label{eqn:mfdjcpb_scheduled}
    \fma_{j,i}^{(t)} = \begin{cases}
        |E_j| \cdot \min(n_j,m_j) & j=i \\[10pt]
        \min_{i \leq k < j} \left( \min \left(
        \begin{aligned}
            &m_j \cdot m_k \cdot n_i + \min \left(
            \begin{aligned}
                &\fma_{j,k+1}^{(t)}+\fma_{k,i}^{(t)} \\
                &\min_{1 \leq t^{*} < t} \left( \max\left(\fma_{j,k+1}^{(t^{*})}, \fma_{k,i}^{(t - t^{*})}\right) \right)
            \end{aligned} \right) \\
            &\fma_{j,k+1}^{(t)} + m_j \cdot \sum_{\nu=i}^k |E_\nu| \\
            &\fma_{k,i}^{(t)} + n_i \cdot \sum_{\nu=k+1}^j |E_\nu|
        \end{aligned}
        \right) \right) & j>i \; .
    \end{cases}
\end{equation}
\begin{flushleft}
    \vfill
    \textbf{\emph{Scheduled} Limited-Memory Matrix-Free Dense Jacobian Chain Product Bracketing}
\end{flushleft}
\begin{equation} \label{eqn:lmmfdjcpb_scheduled}
    \fma_{j,i}^{(t)} = \begin{cases}
        |E_j| \cdot \begin{cases}
            n_j & |E_j| > \overline{M} \\
            \min(n_j,m_j) & \text{otherwise}
        \end{cases} & j=i \\[15pt]
        \min_{i \leq k < j} \left( \min \left(
        \begin{aligned}
            &m_j \cdot m_k \cdot n_i + \min \left(
            \begin{aligned}
                &\fma_{j,k+1}^{(t)}+\fma_{k,i}^{(t)} \\
                &\min_{1 \leq t^{*} < t} \left( \max\left(\fma_{j,k+1}^{(t^{*})}, \fma_{k,i}^{(t - t^{*})}\right) \right)
            \end{aligned} \right) \\
            &\fma_{j,k+1}^{(t)} + m_j \cdot \sum_{\nu=i}^k |E_\nu| \quad \text{if}~\sum_{\nu=i}^k |E_\nu| \leq \overline{M} \\
            &\fma_{k,i}^{(t)} + n_i \cdot \sum_{\nu=k+1}^j |E_\nu|
        \end{aligned}
        \right) \right) & j>i \; .
    \end{cases}
\end{equation}
\begin{flushleft}
    \vfill
\end{flushleft}
%
%
\twocolumn
\subsection{Statistical Results.} \hfill
\label{ssec:appendix_stats}

\begin{figure}[!h]
    \centering
    \input{figures/results_length_3}
    \caption{Ratio between $\fma_{opt}^{(m)}$ and $\fma_{DP}^{(m)}$ calculated for $10^4$ random chains of length $q=3$.}
    \label{fig:results_length_3}
\end{figure}

\begin{figure}[!h]
    \centering
    \input{figures/results_length_4}
    \caption{Ratio between $\fma_{opt}^{(m)}$ and $\fma_{DP}^{(m)}$ calculated for $10^4$ random chains of length $q=4$.}
    \label{fig:results_length_4}
\end{figure}

\begin{figure}[!h]
    \centering
    \input{figures/results_length_5}
    \caption{Ratio between $\fma_{opt}^{(m)}$ and $\fma_{DP}^{(m)}$ calculated for $10^4$ random chains of length $q=5$.}
    \label{fig:results_length_5}
\end{figure}

\newpage
\subsection*{} \hfill

\begin{figure}[!h]
    \centering
    \input{figures/results_length_6}
    \caption{Ratio between $\fma_{opt}^{(m)}$ and $\fma_{DP}^{(m)}$ calculated for $10^4$ random chains of length $q=6$.}
    \label{fig:results_length_6_app}
\end{figure}

\begin{figure}[!h]
    \centering
    \input{figures/results_length_7}
    \caption{Ratio between $\fma_{opt}^{(m)}$ and $\fma_{DP}^{(m)}$ calculated for $10^4$ random chains of length $q=7$.}
    \label{fig:results_length_7}
\end{figure}

\begin{figure}[!h]
    \centering
    \input{figures/results_length_8}
    \caption{Ratio between $\fma_{opt}^{(m)}$ and $\fma_{DP}^{(m)}$ calculated for $10^4$ random chains of length $q=8$.}
    \label{fig:results_length_8}
\end{figure}

%% file: figures/results_length_3.tex
\begingroup%
\makeatletter%
\begin{pgfpicture}%
\pgfpathrectangle{\pgfpointorigin}{\pgfqpoint{3.430000in}{2.040000in}}%
\pgfusepath{use as bounding box, clip}%
\begin{pgfscope}%
\pgfsetbuttcap%
\pgfsetmiterjoin%
\definecolor{currentfill}{rgb}{1.000000,1.000000,1.000000}%
\pgfsetfillcolor{currentfill}%
\pgfsetlinewidth{0.000000pt}%
\definecolor{currentstroke}{rgb}{1.000000,1.000000,1.000000}%
\pgfsetstrokecolor{currentstroke}%
\pgfsetdash{}{0pt}%
\pgfpathmoveto{\pgfqpoint{0.000000in}{0.000000in}}%
\pgfpathlineto{\pgfqpoint{3.430000in}{0.000000in}}%
\pgfpathlineto{\pgfqpoint{3.430000in}{2.040000in}}%
\pgfpathlineto{\pgfqpoint{0.000000in}{2.040000in}}%
\pgfpathlineto{\pgfqpoint{0.000000in}{0.000000in}}%
\pgfpathclose%
\pgfusepath{fill}%
\end{pgfscope}%
\begin{pgfscope}%
\pgfsetbuttcap%
\pgfsetmiterjoin%
\definecolor{currentfill}{rgb}{1.000000,1.000000,1.000000}%
\pgfsetfillcolor{currentfill}%
\pgfsetlinewidth{0.000000pt}%
\definecolor{currentstroke}{rgb}{0.000000,0.000000,0.000000}%
\pgfsetstrokecolor{currentstroke}%
\pgfsetstrokeopacity{0.000000}%
\pgfsetdash{}{0pt}%
\pgfpathmoveto{\pgfqpoint{0.424028in}{0.563372in}}%
\pgfpathlineto{\pgfqpoint{3.280000in}{0.563372in}}%
\pgfpathlineto{\pgfqpoint{3.280000in}{1.683809in}}%
\pgfpathlineto{\pgfqpoint{0.424028in}{1.683809in}}%
\pgfpathlineto{\pgfqpoint{0.424028in}{0.563372in}}%
\pgfpathclose%
\pgfusepath{fill}%
\end{pgfscope}%
\begin{pgfscope}%
\pgfsetbuttcap%
\pgfsetroundjoin%
\definecolor{currentfill}{rgb}{0.000000,0.000000,0.000000}%
\pgfsetfillcolor{currentfill}%
\pgfsetlinewidth{0.803000pt}%
\definecolor{currentstroke}{rgb}{0.000000,0.000000,0.000000}%
\pgfsetstrokecolor{currentstroke}%
\pgfsetdash{}{0pt}%
\pgfsys@defobject{currentmarker}{\pgfqpoint{0.000000in}{-0.048611in}}{\pgfqpoint{0.000000in}{0.000000in}}{%
\pgfpathmoveto{\pgfqpoint{0.000000in}{0.000000in}}%
\pgfpathlineto{\pgfqpoint{0.000000in}{-0.048611in}}%
\pgfusepath{stroke,fill}%
}%
\begin{pgfscope}%
\pgfsys@transformshift{0.900024in}{0.563372in}%
\pgfsys@useobject{currentmarker}{}%
\end{pgfscope}%
\end{pgfscope}%
\begin{pgfscope}%
\definecolor{textcolor}{rgb}{0.000000,0.000000,0.000000}%
\pgfsetstrokecolor{textcolor}%
\pgfsetfillcolor{textcolor}%
\pgftext[x=0.900024in,y=0.466150in,,top]{\color{textcolor}{\rmfamily\fontsize{10.000000}{12.000000}\selectfont\catcode`\^=\active\def^{\ifmmode\sp\else\^{}\fi}\catcode`\%=\active\def
\end{pgfscope}%
\begin{pgfscope}%
\pgfsetbuttcap%
\pgfsetroundjoin%
\definecolor{currentfill}{rgb}{0.000000,0.000000,0.000000}%
\pgfsetfillcolor{currentfill}%
\pgfsetlinewidth{0.803000pt}%
\definecolor{currentstroke}{rgb}{0.000000,0.000000,0.000000}%
\pgfsetstrokecolor{currentstroke}%
\pgfsetdash{}{0pt}%
\pgfsys@defobject{currentmarker}{\pgfqpoint{0.000000in}{-0.048611in}}{\pgfqpoint{0.000000in}{0.000000in}}{%
\pgfpathmoveto{\pgfqpoint{0.000000in}{0.000000in}}%
\pgfpathlineto{\pgfqpoint{0.000000in}{-0.048611in}}%
\pgfusepath{stroke,fill}%
}%
\begin{pgfscope}%
\pgfsys@transformshift{1.852014in}{0.563372in}%
\pgfsys@useobject{currentmarker}{}%
\end{pgfscope}%
\end{pgfscope}%
\begin{pgfscope}%
\definecolor{textcolor}{rgb}{0.000000,0.000000,0.000000}%
\pgfsetstrokecolor{textcolor}%
\pgfsetfillcolor{textcolor}%
\pgftext[x=1.852014in,y=0.466150in,,top]{\color{textcolor}{\rmfamily\fontsize{10.000000}{12.000000}\selectfont\catcode`\^=\active\def^{\ifmmode\sp\else\^{}\fi}\catcode`\%=\active\def
\end{pgfscope}%
\begin{pgfscope}%
\pgfsetbuttcap%
\pgfsetroundjoin%
\definecolor{currentfill}{rgb}{0.000000,0.000000,0.000000}%
\pgfsetfillcolor{currentfill}%
\pgfsetlinewidth{0.803000pt}%
\definecolor{currentstroke}{rgb}{0.000000,0.000000,0.000000}%
\pgfsetstrokecolor{currentstroke}%
\pgfsetdash{}{0pt}%
\pgfsys@defobject{currentmarker}{\pgfqpoint{0.000000in}{-0.048611in}}{\pgfqpoint{0.000000in}{0.000000in}}{%
\pgfpathmoveto{\pgfqpoint{0.000000in}{0.000000in}}%
\pgfpathlineto{\pgfqpoint{0.000000in}{-0.048611in}}%
\pgfusepath{stroke,fill}%
}%
\begin{pgfscope}%
\pgfsys@transformshift{2.804005in}{0.563372in}%
\pgfsys@useobject{currentmarker}{}%
\end{pgfscope}%
\end{pgfscope}%
\begin{pgfscope}%
\definecolor{textcolor}{rgb}{0.000000,0.000000,0.000000}%
\pgfsetstrokecolor{textcolor}%
\pgfsetfillcolor{textcolor}%
\pgftext[x=2.804005in,y=0.466150in,,top]{\color{textcolor}{\rmfamily\fontsize{10.000000}{12.000000}\selectfont\catcode`\^=\active\def^{\ifmmode\sp\else\^{}\fi}\catcode`\%=\active\def
\end{pgfscope}%
\begin{pgfscope}%
\definecolor{textcolor}{rgb}{0.000000,0.000000,0.000000}%
\pgfsetstrokecolor{textcolor}%
\pgfsetfillcolor{textcolor}%
\pgftext[x=1.852014in,y=0.276182in,,top]{\color{textcolor}{\rmfamily\fontsize{10.000000}{12.000000}\selectfont\catcode`\^=\active\def^{\ifmmode\sp\else\^{}\fi}\catcode`\%=\active\def
\end{pgfscope}%
\begin{pgfscope}%
\pgfsetbuttcap%
\pgfsetroundjoin%
\definecolor{currentfill}{rgb}{0.000000,0.000000,0.000000}%
\pgfsetfillcolor{currentfill}%
\pgfsetlinewidth{0.803000pt}%
\definecolor{currentstroke}{rgb}{0.000000,0.000000,0.000000}%
\pgfsetstrokecolor{currentstroke}%
\pgfsetdash{}{0pt}%
\pgfsys@defobject{currentmarker}{\pgfqpoint{-0.048611in}{0.000000in}}{\pgfqpoint{-0.000000in}{0.000000in}}{%
\pgfpathmoveto{\pgfqpoint{-0.000000in}{0.000000in}}%
\pgfpathlineto{\pgfqpoint{-0.048611in}{0.000000in}}%
\pgfusepath{stroke,fill}%
}%
\begin{pgfscope}%
\pgfsys@transformshift{0.424028in}{0.703427in}%
\pgfsys@useobject{currentmarker}{}%
\end{pgfscope}%
\end{pgfscope}%
\begin{pgfscope}%
\definecolor{textcolor}{rgb}{0.000000,0.000000,0.000000}%
\pgfsetstrokecolor{textcolor}%
\pgfsetfillcolor{textcolor}%
\pgftext[x=0.149337in, y=0.650665in, left, base]{\color{textcolor}{\rmfamily\fontsize{10.000000}{12.000000}\selectfont\catcode`\^=\active\def^{\ifmmode\sp\else\^{}\fi}\catcode`\%=\active\def
\end{pgfscope}%
\begin{pgfscope}%
\pgfsetbuttcap%
\pgfsetroundjoin%
\definecolor{currentfill}{rgb}{0.000000,0.000000,0.000000}%
\pgfsetfillcolor{currentfill}%
\pgfsetlinewidth{0.803000pt}%
\definecolor{currentstroke}{rgb}{0.000000,0.000000,0.000000}%
\pgfsetstrokecolor{currentstroke}%
\pgfsetdash{}{0pt}%
\pgfsys@defobject{currentmarker}{\pgfqpoint{-0.048611in}{0.000000in}}{\pgfqpoint{-0.000000in}{0.000000in}}{%
\pgfpathmoveto{\pgfqpoint{-0.000000in}{0.000000in}}%
\pgfpathlineto{\pgfqpoint{-0.048611in}{0.000000in}}%
\pgfusepath{stroke,fill}%
}%
\begin{pgfscope}%
\pgfsys@transformshift{0.424028in}{0.983536in}%
\pgfsys@useobject{currentmarker}{}%
\end{pgfscope}%
\end{pgfscope}%
\begin{pgfscope}%
\definecolor{textcolor}{rgb}{0.000000,0.000000,0.000000}%
\pgfsetstrokecolor{textcolor}%
\pgfsetfillcolor{textcolor}%
\pgftext[x=0.149337in, y=0.930775in, left, base]{\color{textcolor}{\rmfamily\fontsize{10.000000}{12.000000}\selectfont\catcode`\^=\active\def^{\ifmmode\sp\else\^{}\fi}\catcode`\%=\active\def
\end{pgfscope}%
\begin{pgfscope}%
\pgfsetbuttcap%
\pgfsetroundjoin%
\definecolor{currentfill}{rgb}{0.000000,0.000000,0.000000}%
\pgfsetfillcolor{currentfill}%
\pgfsetlinewidth{0.803000pt}%
\definecolor{currentstroke}{rgb}{0.000000,0.000000,0.000000}%
\pgfsetstrokecolor{currentstroke}%
\pgfsetdash{}{0pt}%
\pgfsys@defobject{currentmarker}{\pgfqpoint{-0.048611in}{0.000000in}}{\pgfqpoint{-0.000000in}{0.000000in}}{%
\pgfpathmoveto{\pgfqpoint{-0.000000in}{0.000000in}}%
\pgfpathlineto{\pgfqpoint{-0.048611in}{0.000000in}}%
\pgfusepath{stroke,fill}%
}%
\begin{pgfscope}%
\pgfsys@transformshift{0.424028in}{1.263645in}%
\pgfsys@useobject{currentmarker}{}%
\end{pgfscope}%
\end{pgfscope}%
\begin{pgfscope}%
\definecolor{textcolor}{rgb}{0.000000,0.000000,0.000000}%
\pgfsetstrokecolor{textcolor}%
\pgfsetfillcolor{textcolor}%
\pgftext[x=0.149337in, y=1.210884in, left, base]{\color{textcolor}{\rmfamily\fontsize{10.000000}{12.000000}\selectfont\catcode`\^=\active\def^{\ifmmode\sp\else\^{}\fi}\catcode`\%=\active\def
\end{pgfscope}%
\begin{pgfscope}%
\pgfsetbuttcap%
\pgfsetroundjoin%
\definecolor{currentfill}{rgb}{0.000000,0.000000,0.000000}%
\pgfsetfillcolor{currentfill}%
\pgfsetlinewidth{0.803000pt}%
\definecolor{currentstroke}{rgb}{0.000000,0.000000,0.000000}%
\pgfsetstrokecolor{currentstroke}%
\pgfsetdash{}{0pt}%
\pgfsys@defobject{currentmarker}{\pgfqpoint{-0.048611in}{0.000000in}}{\pgfqpoint{-0.000000in}{0.000000in}}{%
\pgfpathmoveto{\pgfqpoint{-0.000000in}{0.000000in}}%
\pgfpathlineto{\pgfqpoint{-0.048611in}{0.000000in}}%
\pgfusepath{stroke,fill}%
}%
\begin{pgfscope}%
\pgfsys@transformshift{0.424028in}{1.543754in}%
\pgfsys@useobject{currentmarker}{}%
\end{pgfscope}%
\end{pgfscope}%
\begin{pgfscope}%
\definecolor{textcolor}{rgb}{0.000000,0.000000,0.000000}%
\pgfsetstrokecolor{textcolor}%
\pgfsetfillcolor{textcolor}%
\pgftext[x=0.149337in, y=1.490993in, left, base]{\color{textcolor}{\rmfamily\fontsize{10.000000}{12.000000}\selectfont\catcode`\^=\active\def^{\ifmmode\sp\else\^{}\fi}\catcode`\%=\active\def
\end{pgfscope}%
\begin{pgfscope}%
\pgfpathrectangle{\pgfqpoint{0.424028in}{0.563372in}}{\pgfqpoint{2.855972in}{1.120437in}}%
\pgfusepath{clip}%
\pgfsetbuttcap%
\pgfsetroundjoin%
\pgfsetlinewidth{1.505625pt}%
\definecolor{currentstroke}{rgb}{0.501961,0.501961,0.501961}%
\pgfsetstrokecolor{currentstroke}%
\pgfsetdash{{5.550000pt}{2.400000pt}}{0.000000pt}%
\pgfpathmoveto{\pgfqpoint{0.424028in}{1.543754in}}%
\pgfpathlineto{\pgfqpoint{3.280000in}{1.543754in}}%
\pgfusepath{stroke}%
\end{pgfscope}%
\begin{pgfscope}%
\pgfpathrectangle{\pgfqpoint{0.424028in}{0.563372in}}{\pgfqpoint{2.855972in}{1.120437in}}%
\pgfusepath{clip}%
\pgfsetbuttcap%
\pgfsetmiterjoin%
\definecolor{currentfill}{rgb}{0.300000,0.500000,0.700000}%
\pgfsetfillcolor{currentfill}%
\pgfsetfillopacity{0.500000}%
\pgfsetlinewidth{1.003750pt}%
\definecolor{currentstroke}{rgb}{0.248235,0.248235,0.248235}%
\pgfsetstrokecolor{currentstroke}%
\pgfsetdash{}{0pt}%
\pgfpathmoveto{\pgfqpoint{0.662026in}{1.543754in}}%
\pgfpathlineto{\pgfqpoint{1.138021in}{1.543754in}}%
\pgfpathlineto{\pgfqpoint{1.138021in}{1.543754in}}%
\pgfpathlineto{\pgfqpoint{0.662026in}{1.543754in}}%
\pgfpathlineto{\pgfqpoint{0.662026in}{1.543754in}}%
\pgfpathlineto{\pgfqpoint{0.662026in}{1.543754in}}%
\pgfpathclose%
\pgfusepath{stroke,fill}%
\end{pgfscope}%
\begin{pgfscope}%
\pgfpathrectangle{\pgfqpoint{0.424028in}{0.563372in}}{\pgfqpoint{2.855972in}{1.120437in}}%
\pgfusepath{clip}%
\pgfsetbuttcap%
\pgfsetroundjoin%
\pgfsetlinewidth{1.003750pt}%
\definecolor{currentstroke}{rgb}{0.248235,0.248235,0.248235}%
\pgfsetstrokecolor{currentstroke}%
\pgfsetdash{}{0pt}%
\pgfpathmoveto{\pgfqpoint{0.900024in}{1.543754in}}%
\pgfpathlineto{\pgfqpoint{0.900024in}{1.543754in}}%
\pgfusepath{stroke}%
\end{pgfscope}%
\begin{pgfscope}%
\pgfpathrectangle{\pgfqpoint{0.424028in}{0.563372in}}{\pgfqpoint{2.855972in}{1.120437in}}%
\pgfusepath{clip}%
\pgfsetbuttcap%
\pgfsetroundjoin%
\pgfsetlinewidth{1.003750pt}%
\definecolor{currentstroke}{rgb}{0.248235,0.248235,0.248235}%
\pgfsetstrokecolor{currentstroke}%
\pgfsetdash{}{0pt}%
\pgfpathmoveto{\pgfqpoint{0.900024in}{1.543754in}}%
\pgfpathlineto{\pgfqpoint{0.900024in}{1.543754in}}%
\pgfusepath{stroke}%
\end{pgfscope}%
\begin{pgfscope}%
\pgfpathrectangle{\pgfqpoint{0.424028in}{0.563372in}}{\pgfqpoint{2.855972in}{1.120437in}}%
\pgfusepath{clip}%
\pgfsetrectcap%
\pgfsetroundjoin%
\pgfsetlinewidth{1.003750pt}%
\definecolor{currentstroke}{rgb}{0.248235,0.248235,0.248235}%
\pgfsetstrokecolor{currentstroke}%
\pgfsetdash{}{0pt}%
\pgfpathmoveto{\pgfqpoint{0.781025in}{1.543754in}}%
\pgfpathlineto{\pgfqpoint{1.019023in}{1.543754in}}%
\pgfusepath{stroke}%
\end{pgfscope}%
\begin{pgfscope}%
\pgfpathrectangle{\pgfqpoint{0.424028in}{0.563372in}}{\pgfqpoint{2.855972in}{1.120437in}}%
\pgfusepath{clip}%
\pgfsetrectcap%
\pgfsetroundjoin%
\pgfsetlinewidth{1.003750pt}%
\definecolor{currentstroke}{rgb}{0.248235,0.248235,0.248235}%
\pgfsetstrokecolor{currentstroke}%
\pgfsetdash{}{0pt}%
\pgfpathmoveto{\pgfqpoint{0.781025in}{1.543754in}}%
\pgfpathlineto{\pgfqpoint{1.019023in}{1.543754in}}%
\pgfusepath{stroke}%
\end{pgfscope}%
\begin{pgfscope}%
\pgfpathrectangle{\pgfqpoint{0.424028in}{0.563372in}}{\pgfqpoint{2.855972in}{1.120437in}}%
\pgfusepath{clip}%
\pgfsetbuttcap%
\pgfsetmiterjoin%
\definecolor{currentfill}{rgb}{0.300000,0.500000,0.700000}%
\pgfsetfillcolor{currentfill}%
\pgfsetfillopacity{0.500000}%
\pgfsetlinewidth{1.003750pt}%
\definecolor{currentstroke}{rgb}{0.248235,0.248235,0.248235}%
\pgfsetstrokecolor{currentstroke}%
\pgfsetdash{}{0pt}%
\pgfpathmoveto{\pgfqpoint{1.614017in}{1.543754in}}%
\pgfpathlineto{\pgfqpoint{2.090012in}{1.543754in}}%
\pgfpathlineto{\pgfqpoint{2.090012in}{1.543754in}}%
\pgfpathlineto{\pgfqpoint{1.614017in}{1.543754in}}%
\pgfpathlineto{\pgfqpoint{1.614017in}{1.543754in}}%
\pgfpathlineto{\pgfqpoint{1.614017in}{1.543754in}}%
\pgfpathclose%
\pgfusepath{stroke,fill}%
\end{pgfscope}%
\begin{pgfscope}%
\pgfpathrectangle{\pgfqpoint{0.424028in}{0.563372in}}{\pgfqpoint{2.855972in}{1.120437in}}%
\pgfusepath{clip}%
\pgfsetbuttcap%
\pgfsetroundjoin%
\pgfsetlinewidth{1.003750pt}%
\definecolor{currentstroke}{rgb}{0.248235,0.248235,0.248235}%
\pgfsetstrokecolor{currentstroke}%
\pgfsetdash{}{0pt}%
\pgfpathmoveto{\pgfqpoint{1.852014in}{1.543754in}}%
\pgfpathlineto{\pgfqpoint{1.852014in}{1.120910in}}%
\pgfusepath{stroke}%
\end{pgfscope}%
\begin{pgfscope}%
\pgfpathrectangle{\pgfqpoint{0.424028in}{0.563372in}}{\pgfqpoint{2.855972in}{1.120437in}}%
\pgfusepath{clip}%
\pgfsetbuttcap%
\pgfsetroundjoin%
\pgfsetlinewidth{1.003750pt}%
\definecolor{currentstroke}{rgb}{0.248235,0.248235,0.248235}%
\pgfsetstrokecolor{currentstroke}%
\pgfsetdash{}{0pt}%
\pgfpathmoveto{\pgfqpoint{1.852014in}{1.543754in}}%
\pgfpathlineto{\pgfqpoint{1.852014in}{1.543754in}}%
\pgfusepath{stroke}%
\end{pgfscope}%
\begin{pgfscope}%
\pgfpathrectangle{\pgfqpoint{0.424028in}{0.563372in}}{\pgfqpoint{2.855972in}{1.120437in}}%
\pgfusepath{clip}%
\pgfsetrectcap%
\pgfsetroundjoin%
\pgfsetlinewidth{1.003750pt}%
\definecolor{currentstroke}{rgb}{0.248235,0.248235,0.248235}%
\pgfsetstrokecolor{currentstroke}%
\pgfsetdash{}{0pt}%
\pgfpathmoveto{\pgfqpoint{1.733015in}{1.120910in}}%
\pgfpathlineto{\pgfqpoint{1.971013in}{1.120910in}}%
\pgfusepath{stroke}%
\end{pgfscope}%
\begin{pgfscope}%
\pgfpathrectangle{\pgfqpoint{0.424028in}{0.563372in}}{\pgfqpoint{2.855972in}{1.120437in}}%
\pgfusepath{clip}%
\pgfsetrectcap%
\pgfsetroundjoin%
\pgfsetlinewidth{1.003750pt}%
\definecolor{currentstroke}{rgb}{0.248235,0.248235,0.248235}%
\pgfsetstrokecolor{currentstroke}%
\pgfsetdash{}{0pt}%
\pgfpathmoveto{\pgfqpoint{1.733015in}{1.543754in}}%
\pgfpathlineto{\pgfqpoint{1.971013in}{1.543754in}}%
\pgfusepath{stroke}%
\end{pgfscope}%
\begin{pgfscope}%
\pgfpathrectangle{\pgfqpoint{0.424028in}{0.563372in}}{\pgfqpoint{2.855972in}{1.120437in}}%
\pgfusepath{clip}%
\pgfsetbuttcap%
\pgfsetroundjoin%
\definecolor{currentfill}{rgb}{0.000000,0.000000,0.000000}%
\pgfsetfillcolor{currentfill}%
\pgfsetfillopacity{0.000000}%
\pgfsetlinewidth{1.003750pt}%
\definecolor{currentstroke}{rgb}{0.248235,0.248235,0.248235}%
\pgfsetstrokecolor{currentstroke}%
\pgfsetdash{}{0pt}%
\pgfsys@defobject{currentmarker}{\pgfqpoint{-0.041667in}{-0.041667in}}{\pgfqpoint{0.041667in}{0.041667in}}{%
\pgfpathmoveto{\pgfqpoint{-0.041667in}{-0.041667in}}%
\pgfpathlineto{\pgfqpoint{0.041667in}{0.041667in}}%
\pgfpathmoveto{\pgfqpoint{-0.041667in}{0.041667in}}%
\pgfpathlineto{\pgfqpoint{0.041667in}{-0.041667in}}%
\pgfusepath{stroke,fill}%
}%
\begin{pgfscope}%
\pgfsys@transformshift{1.852014in}{0.955678in}%
\pgfsys@useobject{currentmarker}{}%
\end{pgfscope}%
\begin{pgfscope}%
\pgfsys@transformshift{1.852014in}{1.107601in}%
\pgfsys@useobject{currentmarker}{}%
\end{pgfscope}%
\begin{pgfscope}%
\pgfsys@transformshift{1.852014in}{1.044534in}%
\pgfsys@useobject{currentmarker}{}%
\end{pgfscope}%
\begin{pgfscope}%
\pgfsys@transformshift{1.852014in}{1.012982in}%
\pgfsys@useobject{currentmarker}{}%
\end{pgfscope}%
\begin{pgfscope}%
\pgfsys@transformshift{1.852014in}{1.110083in}%
\pgfsys@useobject{currentmarker}{}%
\end{pgfscope}%
\begin{pgfscope}%
\pgfsys@transformshift{1.852014in}{1.018562in}%
\pgfsys@useobject{currentmarker}{}%
\end{pgfscope}%
\begin{pgfscope}%
\pgfsys@transformshift{1.852014in}{1.063819in}%
\pgfsys@useobject{currentmarker}{}%
\end{pgfscope}%
\begin{pgfscope}%
\pgfsys@transformshift{1.852014in}{1.074259in}%
\pgfsys@useobject{currentmarker}{}%
\end{pgfscope}%
\begin{pgfscope}%
\pgfsys@transformshift{1.852014in}{0.868541in}%
\pgfsys@useobject{currentmarker}{}%
\end{pgfscope}%
\begin{pgfscope}%
\pgfsys@transformshift{1.852014in}{1.094440in}%
\pgfsys@useobject{currentmarker}{}%
\end{pgfscope}%
\begin{pgfscope}%
\pgfsys@transformshift{1.852014in}{0.950369in}%
\pgfsys@useobject{currentmarker}{}%
\end{pgfscope}%
\begin{pgfscope}%
\pgfsys@transformshift{1.852014in}{1.064686in}%
\pgfsys@useobject{currentmarker}{}%
\end{pgfscope}%
\begin{pgfscope}%
\pgfsys@transformshift{1.852014in}{0.853194in}%
\pgfsys@useobject{currentmarker}{}%
\end{pgfscope}%
\begin{pgfscope}%
\pgfsys@transformshift{1.852014in}{1.015401in}%
\pgfsys@useobject{currentmarker}{}%
\end{pgfscope}%
\begin{pgfscope}%
\pgfsys@transformshift{1.852014in}{0.902425in}%
\pgfsys@useobject{currentmarker}{}%
\end{pgfscope}%
\begin{pgfscope}%
\pgfsys@transformshift{1.852014in}{1.085871in}%
\pgfsys@useobject{currentmarker}{}%
\end{pgfscope}%
\begin{pgfscope}%
\pgfsys@transformshift{1.852014in}{1.038393in}%
\pgfsys@useobject{currentmarker}{}%
\end{pgfscope}%
\begin{pgfscope}%
\pgfsys@transformshift{1.852014in}{1.112228in}%
\pgfsys@useobject{currentmarker}{}%
\end{pgfscope}%
\begin{pgfscope}%
\pgfsys@transformshift{1.852014in}{1.022155in}%
\pgfsys@useobject{currentmarker}{}%
\end{pgfscope}%
\begin{pgfscope}%
\pgfsys@transformshift{1.852014in}{1.051092in}%
\pgfsys@useobject{currentmarker}{}%
\end{pgfscope}%
\begin{pgfscope}%
\pgfsys@transformshift{1.852014in}{1.103324in}%
\pgfsys@useobject{currentmarker}{}%
\end{pgfscope}%
\begin{pgfscope}%
\pgfsys@transformshift{1.852014in}{1.032472in}%
\pgfsys@useobject{currentmarker}{}%
\end{pgfscope}%
\begin{pgfscope}%
\pgfsys@transformshift{1.852014in}{0.988066in}%
\pgfsys@useobject{currentmarker}{}%
\end{pgfscope}%
\begin{pgfscope}%
\pgfsys@transformshift{1.852014in}{1.110617in}%
\pgfsys@useobject{currentmarker}{}%
\end{pgfscope}%
\begin{pgfscope}%
\pgfsys@transformshift{1.852014in}{1.104139in}%
\pgfsys@useobject{currentmarker}{}%
\end{pgfscope}%
\begin{pgfscope}%
\pgfsys@transformshift{1.852014in}{1.054863in}%
\pgfsys@useobject{currentmarker}{}%
\end{pgfscope}%
\begin{pgfscope}%
\pgfsys@transformshift{1.852014in}{0.998961in}%
\pgfsys@useobject{currentmarker}{}%
\end{pgfscope}%
\begin{pgfscope}%
\pgfsys@transformshift{1.852014in}{0.965633in}%
\pgfsys@useobject{currentmarker}{}%
\end{pgfscope}%
\begin{pgfscope}%
\pgfsys@transformshift{1.852014in}{0.973915in}%
\pgfsys@useobject{currentmarker}{}%
\end{pgfscope}%
\begin{pgfscope}%
\pgfsys@transformshift{1.852014in}{1.072474in}%
\pgfsys@useobject{currentmarker}{}%
\end{pgfscope}%
\begin{pgfscope}%
\pgfsys@transformshift{1.852014in}{1.097052in}%
\pgfsys@useobject{currentmarker}{}%
\end{pgfscope}%
\begin{pgfscope}%
\pgfsys@transformshift{1.852014in}{0.970628in}%
\pgfsys@useobject{currentmarker}{}%
\end{pgfscope}%
\begin{pgfscope}%
\pgfsys@transformshift{1.852014in}{1.116930in}%
\pgfsys@useobject{currentmarker}{}%
\end{pgfscope}%
\begin{pgfscope}%
\pgfsys@transformshift{1.852014in}{0.808564in}%
\pgfsys@useobject{currentmarker}{}%
\end{pgfscope}%
\begin{pgfscope}%
\pgfsys@transformshift{1.852014in}{1.093182in}%
\pgfsys@useobject{currentmarker}{}%
\end{pgfscope}%
\begin{pgfscope}%
\pgfsys@transformshift{1.852014in}{1.110076in}%
\pgfsys@useobject{currentmarker}{}%
\end{pgfscope}%
\begin{pgfscope}%
\pgfsys@transformshift{1.852014in}{0.869512in}%
\pgfsys@useobject{currentmarker}{}%
\end{pgfscope}%
\begin{pgfscope}%
\pgfsys@transformshift{1.852014in}{0.990802in}%
\pgfsys@useobject{currentmarker}{}%
\end{pgfscope}%
\begin{pgfscope}%
\pgfsys@transformshift{1.852014in}{1.105792in}%
\pgfsys@useobject{currentmarker}{}%
\end{pgfscope}%
\begin{pgfscope}%
\pgfsys@transformshift{1.852014in}{0.796661in}%
\pgfsys@useobject{currentmarker}{}%
\end{pgfscope}%
\begin{pgfscope}%
\pgfsys@transformshift{1.852014in}{0.968244in}%
\pgfsys@useobject{currentmarker}{}%
\end{pgfscope}%
\begin{pgfscope}%
\pgfsys@transformshift{1.852014in}{1.059634in}%
\pgfsys@useobject{currentmarker}{}%
\end{pgfscope}%
\begin{pgfscope}%
\pgfsys@transformshift{1.852014in}{0.934371in}%
\pgfsys@useobject{currentmarker}{}%
\end{pgfscope}%
\begin{pgfscope}%
\pgfsys@transformshift{1.852014in}{1.039424in}%
\pgfsys@useobject{currentmarker}{}%
\end{pgfscope}%
\begin{pgfscope}%
\pgfsys@transformshift{1.852014in}{0.830448in}%
\pgfsys@useobject{currentmarker}{}%
\end{pgfscope}%
\begin{pgfscope}%
\pgfsys@transformshift{1.852014in}{1.055110in}%
\pgfsys@useobject{currentmarker}{}%
\end{pgfscope}%
\begin{pgfscope}%
\pgfsys@transformshift{1.852014in}{0.978607in}%
\pgfsys@useobject{currentmarker}{}%
\end{pgfscope}%
\begin{pgfscope}%
\pgfsys@transformshift{1.852014in}{0.892606in}%
\pgfsys@useobject{currentmarker}{}%
\end{pgfscope}%
\begin{pgfscope}%
\pgfsys@transformshift{1.852014in}{1.063650in}%
\pgfsys@useobject{currentmarker}{}%
\end{pgfscope}%
\begin{pgfscope}%
\pgfsys@transformshift{1.852014in}{1.039939in}%
\pgfsys@useobject{currentmarker}{}%
\end{pgfscope}%
\begin{pgfscope}%
\pgfsys@transformshift{1.852014in}{1.030496in}%
\pgfsys@useobject{currentmarker}{}%
\end{pgfscope}%
\begin{pgfscope}%
\pgfsys@transformshift{1.852014in}{0.990249in}%
\pgfsys@useobject{currentmarker}{}%
\end{pgfscope}%
\begin{pgfscope}%
\pgfsys@transformshift{1.852014in}{1.017261in}%
\pgfsys@useobject{currentmarker}{}%
\end{pgfscope}%
\begin{pgfscope}%
\pgfsys@transformshift{1.852014in}{0.885639in}%
\pgfsys@useobject{currentmarker}{}%
\end{pgfscope}%
\begin{pgfscope}%
\pgfsys@transformshift{1.852014in}{1.085102in}%
\pgfsys@useobject{currentmarker}{}%
\end{pgfscope}%
\begin{pgfscope}%
\pgfsys@transformshift{1.852014in}{0.845889in}%
\pgfsys@useobject{currentmarker}{}%
\end{pgfscope}%
\begin{pgfscope}%
\pgfsys@transformshift{1.852014in}{1.059738in}%
\pgfsys@useobject{currentmarker}{}%
\end{pgfscope}%
\begin{pgfscope}%
\pgfsys@transformshift{1.852014in}{1.045019in}%
\pgfsys@useobject{currentmarker}{}%
\end{pgfscope}%
\begin{pgfscope}%
\pgfsys@transformshift{1.852014in}{0.962081in}%
\pgfsys@useobject{currentmarker}{}%
\end{pgfscope}%
\begin{pgfscope}%
\pgfsys@transformshift{1.852014in}{0.829042in}%
\pgfsys@useobject{currentmarker}{}%
\end{pgfscope}%
\begin{pgfscope}%
\pgfsys@transformshift{1.852014in}{1.043663in}%
\pgfsys@useobject{currentmarker}{}%
\end{pgfscope}%
\begin{pgfscope}%
\pgfsys@transformshift{1.852014in}{1.033417in}%
\pgfsys@useobject{currentmarker}{}%
\end{pgfscope}%
\begin{pgfscope}%
\pgfsys@transformshift{1.852014in}{0.997365in}%
\pgfsys@useobject{currentmarker}{}%
\end{pgfscope}%
\begin{pgfscope}%
\pgfsys@transformshift{1.852014in}{0.907653in}%
\pgfsys@useobject{currentmarker}{}%
\end{pgfscope}%
\begin{pgfscope}%
\pgfsys@transformshift{1.852014in}{1.043583in}%
\pgfsys@useobject{currentmarker}{}%
\end{pgfscope}%
\begin{pgfscope}%
\pgfsys@transformshift{1.852014in}{1.011306in}%
\pgfsys@useobject{currentmarker}{}%
\end{pgfscope}%
\begin{pgfscope}%
\pgfsys@transformshift{1.852014in}{0.968111in}%
\pgfsys@useobject{currentmarker}{}%
\end{pgfscope}%
\begin{pgfscope}%
\pgfsys@transformshift{1.852014in}{1.021711in}%
\pgfsys@useobject{currentmarker}{}%
\end{pgfscope}%
\begin{pgfscope}%
\pgfsys@transformshift{1.852014in}{1.091443in}%
\pgfsys@useobject{currentmarker}{}%
\end{pgfscope}%
\begin{pgfscope}%
\pgfsys@transformshift{1.852014in}{1.065762in}%
\pgfsys@useobject{currentmarker}{}%
\end{pgfscope}%
\begin{pgfscope}%
\pgfsys@transformshift{1.852014in}{0.983285in}%
\pgfsys@useobject{currentmarker}{}%
\end{pgfscope}%
\begin{pgfscope}%
\pgfsys@transformshift{1.852014in}{1.022137in}%
\pgfsys@useobject{currentmarker}{}%
\end{pgfscope}%
\begin{pgfscope}%
\pgfsys@transformshift{1.852014in}{1.083927in}%
\pgfsys@useobject{currentmarker}{}%
\end{pgfscope}%
\begin{pgfscope}%
\pgfsys@transformshift{1.852014in}{1.083465in}%
\pgfsys@useobject{currentmarker}{}%
\end{pgfscope}%
\begin{pgfscope}%
\pgfsys@transformshift{1.852014in}{0.977793in}%
\pgfsys@useobject{currentmarker}{}%
\end{pgfscope}%
\begin{pgfscope}%
\pgfsys@transformshift{1.852014in}{1.018992in}%
\pgfsys@useobject{currentmarker}{}%
\end{pgfscope}%
\begin{pgfscope}%
\pgfsys@transformshift{1.852014in}{0.939091in}%
\pgfsys@useobject{currentmarker}{}%
\end{pgfscope}%
\begin{pgfscope}%
\pgfsys@transformshift{1.852014in}{1.053187in}%
\pgfsys@useobject{currentmarker}{}%
\end{pgfscope}%
\begin{pgfscope}%
\pgfsys@transformshift{1.852014in}{0.890235in}%
\pgfsys@useobject{currentmarker}{}%
\end{pgfscope}%
\begin{pgfscope}%
\pgfsys@transformshift{1.852014in}{0.955846in}%
\pgfsys@useobject{currentmarker}{}%
\end{pgfscope}%
\begin{pgfscope}%
\pgfsys@transformshift{1.852014in}{0.901597in}%
\pgfsys@useobject{currentmarker}{}%
\end{pgfscope}%
\begin{pgfscope}%
\pgfsys@transformshift{1.852014in}{0.998426in}%
\pgfsys@useobject{currentmarker}{}%
\end{pgfscope}%
\begin{pgfscope}%
\pgfsys@transformshift{1.852014in}{1.109451in}%
\pgfsys@useobject{currentmarker}{}%
\end{pgfscope}%
\begin{pgfscope}%
\pgfsys@transformshift{1.852014in}{1.005825in}%
\pgfsys@useobject{currentmarker}{}%
\end{pgfscope}%
\begin{pgfscope}%
\pgfsys@transformshift{1.852014in}{0.954670in}%
\pgfsys@useobject{currentmarker}{}%
\end{pgfscope}%
\begin{pgfscope}%
\pgfsys@transformshift{1.852014in}{1.073904in}%
\pgfsys@useobject{currentmarker}{}%
\end{pgfscope}%
\begin{pgfscope}%
\pgfsys@transformshift{1.852014in}{1.078246in}%
\pgfsys@useobject{currentmarker}{}%
\end{pgfscope}%
\begin{pgfscope}%
\pgfsys@transformshift{1.852014in}{1.087495in}%
\pgfsys@useobject{currentmarker}{}%
\end{pgfscope}%
\begin{pgfscope}%
\pgfsys@transformshift{1.852014in}{1.110633in}%
\pgfsys@useobject{currentmarker}{}%
\end{pgfscope}%
\begin{pgfscope}%
\pgfsys@transformshift{1.852014in}{1.022019in}%
\pgfsys@useobject{currentmarker}{}%
\end{pgfscope}%
\begin{pgfscope}%
\pgfsys@transformshift{1.852014in}{1.054853in}%
\pgfsys@useobject{currentmarker}{}%
\end{pgfscope}%
\begin{pgfscope}%
\pgfsys@transformshift{1.852014in}{0.978431in}%
\pgfsys@useobject{currentmarker}{}%
\end{pgfscope}%
\begin{pgfscope}%
\pgfsys@transformshift{1.852014in}{0.974210in}%
\pgfsys@useobject{currentmarker}{}%
\end{pgfscope}%
\begin{pgfscope}%
\pgfsys@transformshift{1.852014in}{1.018899in}%
\pgfsys@useobject{currentmarker}{}%
\end{pgfscope}%
\begin{pgfscope}%
\pgfsys@transformshift{1.852014in}{1.036782in}%
\pgfsys@useobject{currentmarker}{}%
\end{pgfscope}%
\begin{pgfscope}%
\pgfsys@transformshift{1.852014in}{1.068006in}%
\pgfsys@useobject{currentmarker}{}%
\end{pgfscope}%
\begin{pgfscope}%
\pgfsys@transformshift{1.852014in}{0.986758in}%
\pgfsys@useobject{currentmarker}{}%
\end{pgfscope}%
\begin{pgfscope}%
\pgfsys@transformshift{1.852014in}{1.039418in}%
\pgfsys@useobject{currentmarker}{}%
\end{pgfscope}%
\begin{pgfscope}%
\pgfsys@transformshift{1.852014in}{1.024610in}%
\pgfsys@useobject{currentmarker}{}%
\end{pgfscope}%
\begin{pgfscope}%
\pgfsys@transformshift{1.852014in}{0.989611in}%
\pgfsys@useobject{currentmarker}{}%
\end{pgfscope}%
\begin{pgfscope}%
\pgfsys@transformshift{1.852014in}{1.009945in}%
\pgfsys@useobject{currentmarker}{}%
\end{pgfscope}%
\begin{pgfscope}%
\pgfsys@transformshift{1.852014in}{0.871056in}%
\pgfsys@useobject{currentmarker}{}%
\end{pgfscope}%
\begin{pgfscope}%
\pgfsys@transformshift{1.852014in}{1.007628in}%
\pgfsys@useobject{currentmarker}{}%
\end{pgfscope}%
\begin{pgfscope}%
\pgfsys@transformshift{1.852014in}{1.088313in}%
\pgfsys@useobject{currentmarker}{}%
\end{pgfscope}%
\begin{pgfscope}%
\pgfsys@transformshift{1.852014in}{0.915741in}%
\pgfsys@useobject{currentmarker}{}%
\end{pgfscope}%
\begin{pgfscope}%
\pgfsys@transformshift{1.852014in}{1.092326in}%
\pgfsys@useobject{currentmarker}{}%
\end{pgfscope}%
\begin{pgfscope}%
\pgfsys@transformshift{1.852014in}{1.029393in}%
\pgfsys@useobject{currentmarker}{}%
\end{pgfscope}%
\begin{pgfscope}%
\pgfsys@transformshift{1.852014in}{1.118739in}%
\pgfsys@useobject{currentmarker}{}%
\end{pgfscope}%
\begin{pgfscope}%
\pgfsys@transformshift{1.852014in}{0.970620in}%
\pgfsys@useobject{currentmarker}{}%
\end{pgfscope}%
\begin{pgfscope}%
\pgfsys@transformshift{1.852014in}{1.112562in}%
\pgfsys@useobject{currentmarker}{}%
\end{pgfscope}%
\begin{pgfscope}%
\pgfsys@transformshift{1.852014in}{1.079586in}%
\pgfsys@useobject{currentmarker}{}%
\end{pgfscope}%
\begin{pgfscope}%
\pgfsys@transformshift{1.852014in}{0.900212in}%
\pgfsys@useobject{currentmarker}{}%
\end{pgfscope}%
\begin{pgfscope}%
\pgfsys@transformshift{1.852014in}{0.931313in}%
\pgfsys@useobject{currentmarker}{}%
\end{pgfscope}%
\begin{pgfscope}%
\pgfsys@transformshift{1.852014in}{0.982101in}%
\pgfsys@useobject{currentmarker}{}%
\end{pgfscope}%
\begin{pgfscope}%
\pgfsys@transformshift{1.852014in}{0.973912in}%
\pgfsys@useobject{currentmarker}{}%
\end{pgfscope}%
\begin{pgfscope}%
\pgfsys@transformshift{1.852014in}{0.849351in}%
\pgfsys@useobject{currentmarker}{}%
\end{pgfscope}%
\begin{pgfscope}%
\pgfsys@transformshift{1.852014in}{1.091069in}%
\pgfsys@useobject{currentmarker}{}%
\end{pgfscope}%
\begin{pgfscope}%
\pgfsys@transformshift{1.852014in}{0.995056in}%
\pgfsys@useobject{currentmarker}{}%
\end{pgfscope}%
\begin{pgfscope}%
\pgfsys@transformshift{1.852014in}{0.802695in}%
\pgfsys@useobject{currentmarker}{}%
\end{pgfscope}%
\begin{pgfscope}%
\pgfsys@transformshift{1.852014in}{0.984061in}%
\pgfsys@useobject{currentmarker}{}%
\end{pgfscope}%
\begin{pgfscope}%
\pgfsys@transformshift{1.852014in}{1.035949in}%
\pgfsys@useobject{currentmarker}{}%
\end{pgfscope}%
\begin{pgfscope}%
\pgfsys@transformshift{1.852014in}{1.023771in}%
\pgfsys@useobject{currentmarker}{}%
\end{pgfscope}%
\begin{pgfscope}%
\pgfsys@transformshift{1.852014in}{1.102397in}%
\pgfsys@useobject{currentmarker}{}%
\end{pgfscope}%
\begin{pgfscope}%
\pgfsys@transformshift{1.852014in}{0.998690in}%
\pgfsys@useobject{currentmarker}{}%
\end{pgfscope}%
\begin{pgfscope}%
\pgfsys@transformshift{1.852014in}{0.979507in}%
\pgfsys@useobject{currentmarker}{}%
\end{pgfscope}%
\begin{pgfscope}%
\pgfsys@transformshift{1.852014in}{1.006046in}%
\pgfsys@useobject{currentmarker}{}%
\end{pgfscope}%
\begin{pgfscope}%
\pgfsys@transformshift{1.852014in}{1.044536in}%
\pgfsys@useobject{currentmarker}{}%
\end{pgfscope}%
\begin{pgfscope}%
\pgfsys@transformshift{1.852014in}{0.943912in}%
\pgfsys@useobject{currentmarker}{}%
\end{pgfscope}%
\begin{pgfscope}%
\pgfsys@transformshift{1.852014in}{1.037978in}%
\pgfsys@useobject{currentmarker}{}%
\end{pgfscope}%
\begin{pgfscope}%
\pgfsys@transformshift{1.852014in}{1.056032in}%
\pgfsys@useobject{currentmarker}{}%
\end{pgfscope}%
\begin{pgfscope}%
\pgfsys@transformshift{1.852014in}{1.028722in}%
\pgfsys@useobject{currentmarker}{}%
\end{pgfscope}%
\begin{pgfscope}%
\pgfsys@transformshift{1.852014in}{0.963615in}%
\pgfsys@useobject{currentmarker}{}%
\end{pgfscope}%
\begin{pgfscope}%
\pgfsys@transformshift{1.852014in}{0.897082in}%
\pgfsys@useobject{currentmarker}{}%
\end{pgfscope}%
\begin{pgfscope}%
\pgfsys@transformshift{1.852014in}{0.851469in}%
\pgfsys@useobject{currentmarker}{}%
\end{pgfscope}%
\begin{pgfscope}%
\pgfsys@transformshift{1.852014in}{1.071963in}%
\pgfsys@useobject{currentmarker}{}%
\end{pgfscope}%
\begin{pgfscope}%
\pgfsys@transformshift{1.852014in}{1.000837in}%
\pgfsys@useobject{currentmarker}{}%
\end{pgfscope}%
\begin{pgfscope}%
\pgfsys@transformshift{1.852014in}{0.825868in}%
\pgfsys@useobject{currentmarker}{}%
\end{pgfscope}%
\begin{pgfscope}%
\pgfsys@transformshift{1.852014in}{1.110061in}%
\pgfsys@useobject{currentmarker}{}%
\end{pgfscope}%
\begin{pgfscope}%
\pgfsys@transformshift{1.852014in}{1.041565in}%
\pgfsys@useobject{currentmarker}{}%
\end{pgfscope}%
\begin{pgfscope}%
\pgfsys@transformshift{1.852014in}{1.066703in}%
\pgfsys@useobject{currentmarker}{}%
\end{pgfscope}%
\begin{pgfscope}%
\pgfsys@transformshift{1.852014in}{1.058027in}%
\pgfsys@useobject{currentmarker}{}%
\end{pgfscope}%
\begin{pgfscope}%
\pgfsys@transformshift{1.852014in}{1.065289in}%
\pgfsys@useobject{currentmarker}{}%
\end{pgfscope}%
\begin{pgfscope}%
\pgfsys@transformshift{1.852014in}{1.116480in}%
\pgfsys@useobject{currentmarker}{}%
\end{pgfscope}%
\begin{pgfscope}%
\pgfsys@transformshift{1.852014in}{1.007738in}%
\pgfsys@useobject{currentmarker}{}%
\end{pgfscope}%
\begin{pgfscope}%
\pgfsys@transformshift{1.852014in}{1.114311in}%
\pgfsys@useobject{currentmarker}{}%
\end{pgfscope}%
\begin{pgfscope}%
\pgfsys@transformshift{1.852014in}{1.024862in}%
\pgfsys@useobject{currentmarker}{}%
\end{pgfscope}%
\begin{pgfscope}%
\pgfsys@transformshift{1.852014in}{0.919459in}%
\pgfsys@useobject{currentmarker}{}%
\end{pgfscope}%
\begin{pgfscope}%
\pgfsys@transformshift{1.852014in}{0.994818in}%
\pgfsys@useobject{currentmarker}{}%
\end{pgfscope}%
\begin{pgfscope}%
\pgfsys@transformshift{1.852014in}{0.958617in}%
\pgfsys@useobject{currentmarker}{}%
\end{pgfscope}%
\begin{pgfscope}%
\pgfsys@transformshift{1.852014in}{1.092067in}%
\pgfsys@useobject{currentmarker}{}%
\end{pgfscope}%
\begin{pgfscope}%
\pgfsys@transformshift{1.852014in}{1.105671in}%
\pgfsys@useobject{currentmarker}{}%
\end{pgfscope}%
\begin{pgfscope}%
\pgfsys@transformshift{1.852014in}{1.076249in}%
\pgfsys@useobject{currentmarker}{}%
\end{pgfscope}%
\begin{pgfscope}%
\pgfsys@transformshift{1.852014in}{0.989068in}%
\pgfsys@useobject{currentmarker}{}%
\end{pgfscope}%
\begin{pgfscope}%
\pgfsys@transformshift{1.852014in}{1.109352in}%
\pgfsys@useobject{currentmarker}{}%
\end{pgfscope}%
\begin{pgfscope}%
\pgfsys@transformshift{1.852014in}{1.024125in}%
\pgfsys@useobject{currentmarker}{}%
\end{pgfscope}%
\begin{pgfscope}%
\pgfsys@transformshift{1.852014in}{0.873703in}%
\pgfsys@useobject{currentmarker}{}%
\end{pgfscope}%
\begin{pgfscope}%
\pgfsys@transformshift{1.852014in}{0.988173in}%
\pgfsys@useobject{currentmarker}{}%
\end{pgfscope}%
\begin{pgfscope}%
\pgfsys@transformshift{1.852014in}{0.918107in}%
\pgfsys@useobject{currentmarker}{}%
\end{pgfscope}%
\begin{pgfscope}%
\pgfsys@transformshift{1.852014in}{0.996509in}%
\pgfsys@useobject{currentmarker}{}%
\end{pgfscope}%
\begin{pgfscope}%
\pgfsys@transformshift{1.852014in}{1.074724in}%
\pgfsys@useobject{currentmarker}{}%
\end{pgfscope}%
\begin{pgfscope}%
\pgfsys@transformshift{1.852014in}{0.973993in}%
\pgfsys@useobject{currentmarker}{}%
\end{pgfscope}%
\begin{pgfscope}%
\pgfsys@transformshift{1.852014in}{0.977269in}%
\pgfsys@useobject{currentmarker}{}%
\end{pgfscope}%
\begin{pgfscope}%
\pgfsys@transformshift{1.852014in}{1.104904in}%
\pgfsys@useobject{currentmarker}{}%
\end{pgfscope}%
\begin{pgfscope}%
\pgfsys@transformshift{1.852014in}{1.091928in}%
\pgfsys@useobject{currentmarker}{}%
\end{pgfscope}%
\begin{pgfscope}%
\pgfsys@transformshift{1.852014in}{1.119993in}%
\pgfsys@useobject{currentmarker}{}%
\end{pgfscope}%
\begin{pgfscope}%
\pgfsys@transformshift{1.852014in}{1.092287in}%
\pgfsys@useobject{currentmarker}{}%
\end{pgfscope}%
\begin{pgfscope}%
\pgfsys@transformshift{1.852014in}{1.100799in}%
\pgfsys@useobject{currentmarker}{}%
\end{pgfscope}%
\begin{pgfscope}%
\pgfsys@transformshift{1.852014in}{0.900426in}%
\pgfsys@useobject{currentmarker}{}%
\end{pgfscope}%
\begin{pgfscope}%
\pgfsys@transformshift{1.852014in}{0.902441in}%
\pgfsys@useobject{currentmarker}{}%
\end{pgfscope}%
\begin{pgfscope}%
\pgfsys@transformshift{1.852014in}{0.962045in}%
\pgfsys@useobject{currentmarker}{}%
\end{pgfscope}%
\begin{pgfscope}%
\pgfsys@transformshift{1.852014in}{0.955378in}%
\pgfsys@useobject{currentmarker}{}%
\end{pgfscope}%
\begin{pgfscope}%
\pgfsys@transformshift{1.852014in}{1.002067in}%
\pgfsys@useobject{currentmarker}{}%
\end{pgfscope}%
\begin{pgfscope}%
\pgfsys@transformshift{1.852014in}{1.044558in}%
\pgfsys@useobject{currentmarker}{}%
\end{pgfscope}%
\begin{pgfscope}%
\pgfsys@transformshift{1.852014in}{1.035701in}%
\pgfsys@useobject{currentmarker}{}%
\end{pgfscope}%
\begin{pgfscope}%
\pgfsys@transformshift{1.852014in}{1.113615in}%
\pgfsys@useobject{currentmarker}{}%
\end{pgfscope}%
\begin{pgfscope}%
\pgfsys@transformshift{1.852014in}{1.031685in}%
\pgfsys@useobject{currentmarker}{}%
\end{pgfscope}%
\begin{pgfscope}%
\pgfsys@transformshift{1.852014in}{1.108048in}%
\pgfsys@useobject{currentmarker}{}%
\end{pgfscope}%
\begin{pgfscope}%
\pgfsys@transformshift{1.852014in}{0.958437in}%
\pgfsys@useobject{currentmarker}{}%
\end{pgfscope}%
\begin{pgfscope}%
\pgfsys@transformshift{1.852014in}{1.114361in}%
\pgfsys@useobject{currentmarker}{}%
\end{pgfscope}%
\begin{pgfscope}%
\pgfsys@transformshift{1.852014in}{0.939257in}%
\pgfsys@useobject{currentmarker}{}%
\end{pgfscope}%
\begin{pgfscope}%
\pgfsys@transformshift{1.852014in}{1.033951in}%
\pgfsys@useobject{currentmarker}{}%
\end{pgfscope}%
\begin{pgfscope}%
\pgfsys@transformshift{1.852014in}{1.048388in}%
\pgfsys@useobject{currentmarker}{}%
\end{pgfscope}%
\begin{pgfscope}%
\pgfsys@transformshift{1.852014in}{1.089988in}%
\pgfsys@useobject{currentmarker}{}%
\end{pgfscope}%
\begin{pgfscope}%
\pgfsys@transformshift{1.852014in}{1.112654in}%
\pgfsys@useobject{currentmarker}{}%
\end{pgfscope}%
\begin{pgfscope}%
\pgfsys@transformshift{1.852014in}{1.107364in}%
\pgfsys@useobject{currentmarker}{}%
\end{pgfscope}%
\begin{pgfscope}%
\pgfsys@transformshift{1.852014in}{0.820788in}%
\pgfsys@useobject{currentmarker}{}%
\end{pgfscope}%
\begin{pgfscope}%
\pgfsys@transformshift{1.852014in}{1.104359in}%
\pgfsys@useobject{currentmarker}{}%
\end{pgfscope}%
\begin{pgfscope}%
\pgfsys@transformshift{1.852014in}{1.029418in}%
\pgfsys@useobject{currentmarker}{}%
\end{pgfscope}%
\begin{pgfscope}%
\pgfsys@transformshift{1.852014in}{1.094062in}%
\pgfsys@useobject{currentmarker}{}%
\end{pgfscope}%
\begin{pgfscope}%
\pgfsys@transformshift{1.852014in}{1.029272in}%
\pgfsys@useobject{currentmarker}{}%
\end{pgfscope}%
\begin{pgfscope}%
\pgfsys@transformshift{1.852014in}{1.032722in}%
\pgfsys@useobject{currentmarker}{}%
\end{pgfscope}%
\begin{pgfscope}%
\pgfsys@transformshift{1.852014in}{0.966647in}%
\pgfsys@useobject{currentmarker}{}%
\end{pgfscope}%
\begin{pgfscope}%
\pgfsys@transformshift{1.852014in}{1.089326in}%
\pgfsys@useobject{currentmarker}{}%
\end{pgfscope}%
\begin{pgfscope}%
\pgfsys@transformshift{1.852014in}{0.879162in}%
\pgfsys@useobject{currentmarker}{}%
\end{pgfscope}%
\begin{pgfscope}%
\pgfsys@transformshift{1.852014in}{1.039132in}%
\pgfsys@useobject{currentmarker}{}%
\end{pgfscope}%
\begin{pgfscope}%
\pgfsys@transformshift{1.852014in}{1.085953in}%
\pgfsys@useobject{currentmarker}{}%
\end{pgfscope}%
\begin{pgfscope}%
\pgfsys@transformshift{1.852014in}{1.088146in}%
\pgfsys@useobject{currentmarker}{}%
\end{pgfscope}%
\begin{pgfscope}%
\pgfsys@transformshift{1.852014in}{1.084029in}%
\pgfsys@useobject{currentmarker}{}%
\end{pgfscope}%
\begin{pgfscope}%
\pgfsys@transformshift{1.852014in}{1.099861in}%
\pgfsys@useobject{currentmarker}{}%
\end{pgfscope}%
\begin{pgfscope}%
\pgfsys@transformshift{1.852014in}{0.965227in}%
\pgfsys@useobject{currentmarker}{}%
\end{pgfscope}%
\end{pgfscope}%
\begin{pgfscope}%
\pgfpathrectangle{\pgfqpoint{0.424028in}{0.563372in}}{\pgfqpoint{2.855972in}{1.120437in}}%
\pgfusepath{clip}%
\pgfsetbuttcap%
\pgfsetmiterjoin%
\definecolor{currentfill}{rgb}{0.300000,0.500000,0.700000}%
\pgfsetfillcolor{currentfill}%
\pgfsetfillopacity{0.500000}%
\pgfsetlinewidth{1.003750pt}%
\definecolor{currentstroke}{rgb}{0.248235,0.248235,0.248235}%
\pgfsetstrokecolor{currentstroke}%
\pgfsetdash{}{0pt}%
\pgfpathmoveto{\pgfqpoint{2.566007in}{1.543754in}}%
\pgfpathlineto{\pgfqpoint{3.042002in}{1.543754in}}%
\pgfpathlineto{\pgfqpoint{3.042002in}{1.543754in}}%
\pgfpathlineto{\pgfqpoint{2.566007in}{1.543754in}}%
\pgfpathlineto{\pgfqpoint{2.566007in}{1.543754in}}%
\pgfpathlineto{\pgfqpoint{2.566007in}{1.543754in}}%
\pgfpathclose%
\pgfusepath{stroke,fill}%
\end{pgfscope}%
\begin{pgfscope}%
\pgfpathrectangle{\pgfqpoint{0.424028in}{0.563372in}}{\pgfqpoint{2.855972in}{1.120437in}}%
\pgfusepath{clip}%
\pgfsetbuttcap%
\pgfsetroundjoin%
\pgfsetlinewidth{1.003750pt}%
\definecolor{currentstroke}{rgb}{0.248235,0.248235,0.248235}%
\pgfsetstrokecolor{currentstroke}%
\pgfsetdash{}{0pt}%
\pgfpathmoveto{\pgfqpoint{2.804005in}{1.543754in}}%
\pgfpathlineto{\pgfqpoint{2.804005in}{1.543754in}}%
\pgfusepath{stroke}%
\end{pgfscope}%
\begin{pgfscope}%
\pgfpathrectangle{\pgfqpoint{0.424028in}{0.563372in}}{\pgfqpoint{2.855972in}{1.120437in}}%
\pgfusepath{clip}%
\pgfsetbuttcap%
\pgfsetroundjoin%
\pgfsetlinewidth{1.003750pt}%
\definecolor{currentstroke}{rgb}{0.248235,0.248235,0.248235}%
\pgfsetstrokecolor{currentstroke}%
\pgfsetdash{}{0pt}%
\pgfpathmoveto{\pgfqpoint{2.804005in}{1.543754in}}%
\pgfpathlineto{\pgfqpoint{2.804005in}{1.543754in}}%
\pgfusepath{stroke}%
\end{pgfscope}%
\begin{pgfscope}%
\pgfpathrectangle{\pgfqpoint{0.424028in}{0.563372in}}{\pgfqpoint{2.855972in}{1.120437in}}%
\pgfusepath{clip}%
\pgfsetrectcap%
\pgfsetroundjoin%
\pgfsetlinewidth{1.003750pt}%
\definecolor{currentstroke}{rgb}{0.248235,0.248235,0.248235}%
\pgfsetstrokecolor{currentstroke}%
\pgfsetdash{}{0pt}%
\pgfpathmoveto{\pgfqpoint{2.685006in}{1.543754in}}%
\pgfpathlineto{\pgfqpoint{2.923004in}{1.543754in}}%
\pgfusepath{stroke}%
\end{pgfscope}%
\begin{pgfscope}%
\pgfpathrectangle{\pgfqpoint{0.424028in}{0.563372in}}{\pgfqpoint{2.855972in}{1.120437in}}%
\pgfusepath{clip}%
\pgfsetrectcap%
\pgfsetroundjoin%
\pgfsetlinewidth{1.003750pt}%
\definecolor{currentstroke}{rgb}{0.248235,0.248235,0.248235}%
\pgfsetstrokecolor{currentstroke}%
\pgfsetdash{}{0pt}%
\pgfpathmoveto{\pgfqpoint{2.685006in}{1.543754in}}%
\pgfpathlineto{\pgfqpoint{2.923004in}{1.543754in}}%
\pgfusepath{stroke}%
\end{pgfscope}%
\begin{pgfscope}%
\pgfpathrectangle{\pgfqpoint{0.424028in}{0.563372in}}{\pgfqpoint{2.855972in}{1.120437in}}%
\pgfusepath{clip}%
\pgfsetbuttcap%
\pgfsetroundjoin%
\pgfsetlinewidth{1.505625pt}%
\definecolor{currentstroke}{rgb}{1.000000,0.647059,0.000000}%
\pgfsetstrokecolor{currentstroke}%
\pgfsetdash{}{0pt}%
\pgfpathmoveto{\pgfqpoint{0.662026in}{1.543754in}}%
\pgfpathlineto{\pgfqpoint{1.138021in}{1.543754in}}%
\pgfusepath{stroke}%
\end{pgfscope}%
\begin{pgfscope}%
\pgfpathrectangle{\pgfqpoint{0.424028in}{0.563372in}}{\pgfqpoint{2.855972in}{1.120437in}}%
\pgfusepath{clip}%
\pgfsetbuttcap%
\pgfsetroundjoin%
\pgfsetlinewidth{1.505625pt}%
\definecolor{currentstroke}{rgb}{1.000000,0.647059,0.000000}%
\pgfsetstrokecolor{currentstroke}%
\pgfsetdash{}{0pt}%
\pgfpathmoveto{\pgfqpoint{1.614017in}{1.543754in}}%
\pgfpathlineto{\pgfqpoint{2.090012in}{1.543754in}}%
\pgfusepath{stroke}%
\end{pgfscope}%
\begin{pgfscope}%
\pgfpathrectangle{\pgfqpoint{0.424028in}{0.563372in}}{\pgfqpoint{2.855972in}{1.120437in}}%
\pgfusepath{clip}%
\pgfsetbuttcap%
\pgfsetroundjoin%
\pgfsetlinewidth{1.505625pt}%
\definecolor{currentstroke}{rgb}{1.000000,0.647059,0.000000}%
\pgfsetstrokecolor{currentstroke}%
\pgfsetdash{}{0pt}%
\pgfpathmoveto{\pgfqpoint{2.566007in}{1.543754in}}%
\pgfpathlineto{\pgfqpoint{3.042002in}{1.543754in}}%
\pgfusepath{stroke}%
\end{pgfscope}%
\begin{pgfscope}%
\pgfsetrectcap%
\pgfsetmiterjoin%
\pgfsetlinewidth{0.803000pt}%
\definecolor{currentstroke}{rgb}{0.000000,0.000000,0.000000}%
\pgfsetstrokecolor{currentstroke}%
\pgfsetdash{}{0pt}%
\pgfpathmoveto{\pgfqpoint{0.424028in}{0.563372in}}%
\pgfpathlineto{\pgfqpoint{0.424028in}{1.683809in}}%
\pgfusepath{stroke}%
\end{pgfscope}%
\begin{pgfscope}%
\pgfsetrectcap%
\pgfsetmiterjoin%
\pgfsetlinewidth{0.803000pt}%
\definecolor{currentstroke}{rgb}{0.000000,0.000000,0.000000}%
\pgfsetstrokecolor{currentstroke}%
\pgfsetdash{}{0pt}%
\pgfpathmoveto{\pgfqpoint{3.280000in}{0.563372in}}%
\pgfpathlineto{\pgfqpoint{3.280000in}{1.683809in}}%
\pgfusepath{stroke}%
\end{pgfscope}%
\begin{pgfscope}%
\pgfsetrectcap%
\pgfsetmiterjoin%
\pgfsetlinewidth{0.803000pt}%
\definecolor{currentstroke}{rgb}{0.000000,0.000000,0.000000}%
\pgfsetstrokecolor{currentstroke}%
\pgfsetdash{}{0pt}%
\pgfpathmoveto{\pgfqpoint{0.424028in}{0.563372in}}%
\pgfpathlineto{\pgfqpoint{3.280000in}{0.563372in}}%
\pgfusepath{stroke}%
\end{pgfscope}%
\begin{pgfscope}%
\pgfsetrectcap%
\pgfsetmiterjoin%
\pgfsetlinewidth{0.803000pt}%
\definecolor{currentstroke}{rgb}{0.000000,0.000000,0.000000}%
\pgfsetstrokecolor{currentstroke}%
\pgfsetdash{}{0pt}%
\pgfpathmoveto{\pgfqpoint{0.424028in}{1.683809in}}%
\pgfpathlineto{\pgfqpoint{3.280000in}{1.683809in}}%
\pgfusepath{stroke}%
\end{pgfscope}%
\begin{pgfscope}%
\definecolor{textcolor}{rgb}{0.000000,0.000000,0.000000}%
\pgfsetstrokecolor{textcolor}%
\pgfsetfillcolor{textcolor}%
\pgftext[x=1.852014in,y=1.767142in,,base]{\color{textcolor}{\rmfamily\fontsize{10.000000}{12.000000}\selectfont\catcode`\^=\active\def^{\ifmmode\sp\else\^{}\fi}\catcode`\%=\active\def
\end{pgfscope}%
\end{pgfpicture}%
\makeatother%
\endgroup%